\documentclass[12pt,draftclsnofoot,onecolumn]{IEEEtran}
\IEEEoverridecommandlockouts
\usepackage{float}
\usepackage{placeins}
\usepackage[space]{cite}
\usepackage[utf8]{inputenc}
\usepackage[T1]{fontenc}
\usepackage{amsmath}
\usepackage{mathtools}
\usepackage{array}
\usepackage{multirow}
\newcommand{\PreserveBackslash}[1]{\let\temp=\\#1\let\\=\temp}
\newcolumntype{C}[1]{>{\PreserveBackslash\centering}p{#1}}
\newcolumntype{R}[1]{>{\PreserveBackslash\raggedleft}p{#1}}
\newcolumntype{L}[1]{>{\PreserveBackslash\raggedright}p{#1}}

\usepackage[colorinlistoftodos]{todonotes}

\usepackage[noend]{algorithmic}
\usepackage[linesnumbered,ruled,vlined,boxed,commentsnumbered]{algorithm2e}
\usepackage{amsmath}
\usepackage{pgfplots}
\usepackage{graphicx} 
\usepackage{subfigure}
\usepackage{multirow}
\usepackage{paralist}
\usepackage{wrapfig}
\usepackage{url}
\usepackage{amsfonts}
\usepackage{amssymb}
\usepackage{tikz}
\usetikzlibrary{scopes,patterns,intersections,calc}
\usepackage{cancel}
\usepackage{bm}
\usepackage{nicefrac}
\usepackage{color}
\usepackage{diagbox}
\usepackage{enumerate}
\usepackage{booktabs}
\usepackage{colortbl}
\usepackage{tabularx}
\usepackage{epstopdf}
\usepackage{diagbox}

\newtheorem{definition}{\textbf{Definition}}

\newtheorem{theorem}{\textbf{Theorem}}
\newtheorem{lemma}{\textbf{Lemma}}

\newcommand{\remove}[1]{}

\newenvironment{proof}{\noindent {\bf
		Proof.}}{\rule{3mm}{3mm}\par\medskip}

\pgfplotsset{compat=1.3}

  \definecolor{nvb}{RGB}{65,105,225}
\newcommand*\circled[1]{\tikz[baseline=(char.base)]{
            \node[shape=circle,draw,inner sep=1pt] (char) {#1};}}
            
\newcolumntype{P}[1]{>{\centering\arraybackslash}p{#1}}
\newcolumntype{M}[1]{>{\centering\arraybackslash}m{#1}}
\begin{document}
	
\title{On Virtual Network Embedding: Paths and Cycles}	

 \author{Haitao~Wu,~Fen~Zhou,~\IEEEmembership{Senior~Member,~IEEE},~Yaojun~Chen,~Ran~Zhang

 \thanks{H. Wu is with the Department of mathematics, Nanjing University, China. He is also with the CERI-LIA at the University of Avignon. (email: whtmath2011@163.com).}
 \thanks{F. Zhou is with the LISITE lab of the Institut Superieur d'Electronique de Paris. He is also with the CERI-LIA at the University of Avignon. (emails: fen.zhou@(isep,  univ-avignon).fr).}
 \thanks{Y. Chen is with Department of Mathematics, Nanjing University, Nanjing 210093, China. (email: yaojunc@nju.edu.cn).}
  \thanks{R. Zhang is with School of Mathematics, Shanghai University of Finance and Economics, Guoding Road 777, 200433, Shanghai, P.R.China. (email: zhang.ran@mail.shufe.edu.cn).}}

\maketitle

\begin{abstract}
Network virtualization provides a promising solution to overcome the ossification of current networks, allowing multiple Virtual Network Requests (VNRs) embedded on a common infrastructure. The major challenge in network virtualization is the Virtual Network Embedding (VNE) problem, which is to embed VNRs onto a shared substrate network and known to be $\mathcal{NP}$-hard. The topological heterogeneity of VNRs is one important factor hampering the performance of the VNE. However, in many specialized applications and infrastructures, VNRs are of some common structural features $\textit{e.g.}$, paths and cycles. To achieve better outcomes, it is thus critical to design dedicated algorithms for these applications and infrastructures by taking into accounting topological characteristics. Besides, paths and cycles are two of the most fundamental topologies that all network structures consist of. Exploiting the characteristics of path and cycle embeddings is vital to tackle the general VNE problem. In this paper, we investigated the path and cycle embedding problems. For path embedding, we proved its $\mathcal{NP}$-hardness and inapproximability. Then, by utilizing Multiple Knapsack Problem (MKP) and Multi-Dimensional Knapsack Problem (MDKP), we proposed an efficient and effective MKP-MDKP-based algorithm. For cycle embedding, we proposed a Weighted Directed Auxiliary Graph (WDAG) to develop a polynomial-time algorithm to determine the least-resource-consuming embedding. Numerical results showed our customized algorithms can boost the acceptance ratio and revenue compared to generic embedding algorithms in the literature.
\end{abstract}

\section{Introduction}
\label{sec:intro}
\IEEEPARstart{N}{owadays}, the trend of Internet, especially driven by Big data applications \cite{s1:b1,s1:b2}, marches towards involving more network elements and end-users, larger volume
of traffic, and more diversified applications. However, current Internet infrastructures, consisting of a variety of technologies to run distributed protocols, become barrier to  improving Internet service. This diversification is often referred to as the Internet ossification problem \cite{s1:b3}. Network virtualization has been regarded as a compelling tool to overcome the Internet ossification and attracting a lot of researches \cite{s1:b4,s1:b5,s1:b6,s1:b7,s1:b8,s1:b9}. It supports various networks of diverse natures ($\textit{e.g.}$, network architectures, protocols, and user interactions \cite{s1:b4}) to coexist in a same substrate network and share substrate resources ($\textit{e.g.}$, CPUs and bandwidths). In the paradigm of network virtualization, the role of traditional Internet Service Providers (ISPs) is separated into two new entities: Infrastructure Provider (InP) and Service Provider (SP). The InP owns and manages the substrate network while the SP focuses on offering customized services to clients. In this business model as shown in Fig. \ref{fig:f1},
\begin{figure}[!htb]
\centering
\includegraphics[height=0.45\textwidth]{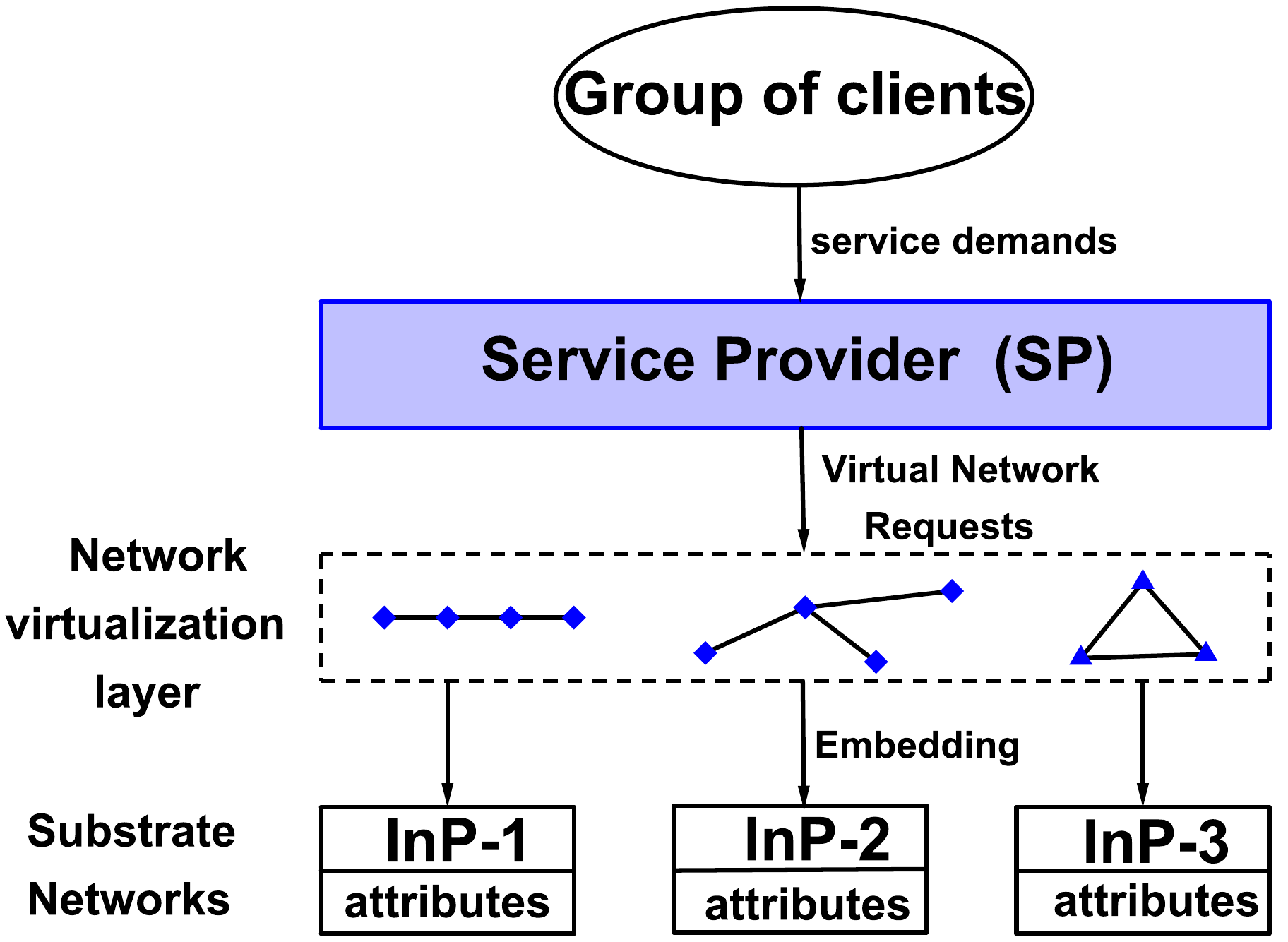}
\caption{The paradigm of network virtualization.}
\label{fig:f1}
\end{figure}
the InP sets up and maintains the physical equipments and substrate resources such as optical fibers, CPUs and bandwidths as well as network protocols. Herein, these physical equipments and resources compose the attributes of the InP, which serve to discover resources for SPs \cite{s1:b6}. The SP, pursuant to service demands of clients, creates Virtual Network Requests (VNR) (A VNR is a combination of Virtual Nodes (VNs) and Virtual Links (VLs) \cite{s1:b5}). It then discovers resources available in substrate networks by the attributes of InPs and selects appropriate ones for the deployment of VNRs \cite{s1:b6}.

How to effectively allocate resources of the substrate network to VNRs is a vital problem in network virtualization, which is often referred to as the Virtual Network Embedding (VNE) problem \cite{s1:b5}. Explicitly, the VNE needs to (a) find a Substrate Node (SN) to meet the computing requirement of each VN, and (b) find a substrate path to satisfy the bandwidth requirement of each VL in a VNR. The former is also called $\textit{Node Mapping}$ and the latter is named $\textit{Link Mapping}$. The VNE has been proven $\mathcal{NP}$-hard \cite{s1:b10} and studied intensively  \cite{s1:b4,s1:b7,s1:b8,s1:b9,s1:b11,s1:b12,s1:b13,s1:b14,S6:b2}. These works introduce different methods like heuristic algorithms and Integer Linear Programming (ILP) models, \textit{etc}, and cover many aspects, such as distributed computing of the VNE and embedding across multiple substrate networks.

One of the key impediments in the general VNE problem is the topological heterogeneity of both VNRs and substrate networks \cite{s1:b6}. However, this is not always true in many specific applications and substrate networks. For instance, the topologies of network service chains  are paths \cite{s1:b15}, and there are many substrate optical rings (\textit{i.e.}, cycles) \cite{s1:b16}. For these applications and infrastructures, specialized cloud service providers outperforming the general SPs are desired, where dedicated algorithms, taking into account the topological characteristics of the VNRs and substrate networks, can be afforded. Besides, paths and cycles are two of the most fundamental topologies in network structures. Exploiting the characteristics of path and cycle embeddings is vital to tackle general topology embedding. 

For example, if path and cycle embeddings can be effectively solved, we can decompose a general VNR into paths and cycles and then embed them on the specialized platforms, as shown in Fig. \ref{fig:f2}, to boost the performance of the VNE. (The feasibility of embedding a VNR across multiple substrate networks has been verified in \cite{s1:b9}, which makes the idea of decomposition practicable.)

 \begin{figure}[!htb]
 \centering
 \includegraphics[height=0.25\textwidth]{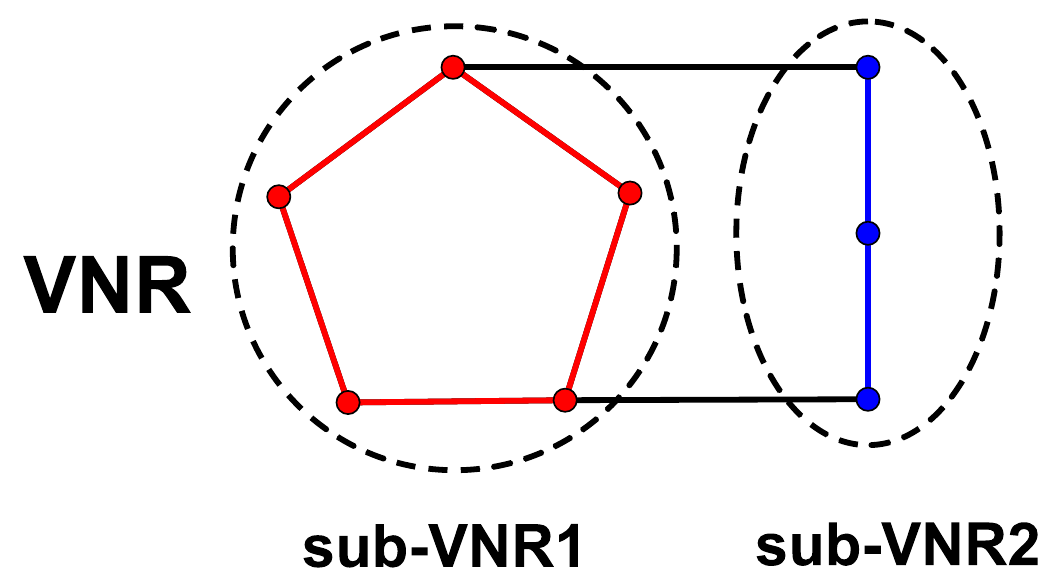}
 \caption{A general VNR decomposed into two sub-VNRs, one cycle and one path.}
 \label{fig:f2}
 \end{figure}

In the wake of the idea of the two special embeddings (paths and cycles), there are some important questions not answered yet: How hard are they? Still $\mathcal{NP}$-hard or there exist polynomial-time algorithms to solve them exactly or approximately? In this paper, we comprehensively investigate path and cycle embeddings from theoretical hardness analysis to practical algorithm design. The main contributions of this work are summarized as follows.
\begin{itemize}
\item From the theoretical perspective, we formally define the decision problem $Emb(G^s,G^r)$ of a single VNR $G^r$ embedded onto a substrate network $G^s$. For path embedding, we provided a polynomial-time algorithm to optimally solve the path-to-path embedding. For general topology, we proved $Emb(G^s, P^r)$ is $\mathcal{NP}$-hard and further analyzed its inapproximability. For cycle Embedding, we proved that $Emb(C^s, C^r)$ can be polynomial-timely solved while the problems of maximizing acceptance ratio and revenue are still strongly $\mathcal{NP}$-hard.  
\item From the perspective of practical algorithm design, customized   algorithms have been proposed by taking into account the nature of each embedding type. Following the idea of decomposing a substrate network into 'paths',  we transformed this problem into a Multiple Knapsack Problem (MKP) and Multi-Dimensional Knapsack Problem (MDKP), and developed an efficient and effective MKP-MDKP-based algorithm to solve it. For cycle embeding, we proposed a Weighted Directed Auxiliary Graph (WDAG) and succeeded to establish an one-to-one relation between each directed cycle in WDAG and each feasible embedding. Based on that, a polynomial-time algorithm is herein devised to achieve the least-resource-consuming embedding. To the best of our knowledge, it is the first time that the idea of embedding decomposition is proposed. 
\end{itemize}


The rest of this paper is organized as follows. Section \ref{sec:rw} briefly introduces the related work and our motivation. We present the network models and the formal description of the VNE problem in Section \ref{sec:nmandpd}. Then for the theoretical perspective of path embedding, we provide the proof of $\mathcal{NP}$-hardness and inapproximability in Section \ref{sec:pm}. For the practical algorithm, we present our design framework and devise the MKP-MDKP-based algorithms in Section \ref{sec:rm}. For cycle embedding, Section \ref{sec:ce} elaborates the construction of WDAG, characterizes the one-to-one relation between directed cycles and feasible embedding ways, and further devises the specialized cycle-embedding algorithm. We conduct simulations under different scenarios in Section \ref{sec:nr} to demonstrate the superiority of our proposed algorithms over the existing general algorithms. Finally, Section \ref{sec:conclusion} summarizes this paper.

\section{Related Work and Motivation}
\label{sec:rw}
The VNE, as the main challenging problem in network virtualization, drew a lot of attentions of researchers. In \cite{s1:b5}, the authors expanded the roles of the SP and InP in the paradigm of network virtualization and proposed a novel classification scheme for current VNE algorithms. Another comprehensive survey \cite{s1:b6} elaborated and emphasized the importance of resource discovery and allocation of the VNE. Many solutions to the VNE problem have been proposed in the literature \cite{s1:b4,s1:b7,s1:b8,s1:b9,s1:b11,s1:b12,s1:b13,s1:b14,S6:b2} including heuristic-based,  ILP, $\textit{etc}$. Later in \cite{s1:b4,s1:b7}, researchers found that the topology information of VNRs and substrate networks can be utilized to improve the performance of the VNE. The authors of \cite{s1:b7} applied a Markov random walk model, analogous to the idea of PageRank \cite{s2:b2}, to rank network nodes based on its resource and topological attributes. In \cite{s1:b4}, customized embedding algorithms for some special classes of topologies have been investigated and proven more effective than the general algorithms. Although VNRs may have arbitrary topologies, the network structures of some key applications and infrastructures are of common topologies $\textit{e.g.}$, paths and cycles \cite{s1:b15,s1:b16}. However, few relevant works intentionally pay attention to the two special but relatively common topologies in the VNE problem. Besides, paths and cycles are two of the most fundamental topologies in network structures. Since the general VNE problem is computationally hard, it is a pragmatic way to decompose VNRs into several specific substructures of paths and cycles and then effectively embed them separately. Embedding across multiple substrate networks and distributed embeddings of the VNE have been shown to be feasible \cite{s1:b5,s1:b9}. This makes the idea of decomposing VNRs and embedding separately practicable. Therefore, a devoted study to explore the characteristics of path and cycle embeddings is desired, which has not yet been researched. 

In this work, we shall systematically, from theoretical hardness analysis to practical algorithm design, investigate the VNE problem for the two special topologies.

\section{Network Models and Problem Description}
\label{sec:nmandpd}
In this section, we first present the network models considered in this paper and then give the formulation of the VNE. Some necessary notations are summarized in Table \ref{tab: notations}.
\begin{table}[!htp]
 \centering
\caption{Notations} \label{tab: notations}
\begin{tabular}{|m{2cm}|m{12cm}|}
\hline
\textbf{Notation} & \textbf{Description}\\
\hline
~~~~$G^s(V^s,E^s)$ & A substrate network, where $V^s$ is the set of SNs, and $E^s$ is the set of Substrate Links (SLs)\\
\hline
~~~~~~~~~$v^s$  &$v^s\in V^s$, a Substrate Node (SN)\\
\hline
~~~~~~~~~$e^s$ & $e^s \in E^s$, a Substrate Link (SL)\\ 
\hline
~~~~~~~$v^s_iv^s_j$ & $v^s_iv^s_j \in E^s$, the SL connecting $v^s_i \in V^s$ and $v^s_j \in V^s$\\ 
\hline
~~~~~~$\mathcal{P}_{v^s_iv^s_j}$ & The set of all substrate paths from $v^s_i$ to $v^s_j$ in $G^s$\\
\hline
~~~~~~$P_{v^s_iv^s_j}$ & $P_{v^s_iv^s_j} \in \mathcal{P}_{v^s_iv^s_j}$, a substrate path from $v^s_i$ to $v^s_j$\\
\hline
~~~~~~$P^s$ & A substrate network which is a path with $|P^s|$ SLs on it\\
\hline
$C^s(v^s_1,...,v^s_mv^s_1)$ & A substrate cycle with clockwise order of SNs $\textit{i.e.}$, starting from $v^s_1$ clockwise to $v^s_m$, where $m$ is the number of SNs on $C^s$ \\
\hline
\hline
~~~~$G^r(V^r,E^r)$  &A Virtual Network Request (VNR), where $V^r$ is the set of VNs and $E^r$ is the set of VLs\\
\hline
~~~~~~~~$v^r$  &$v^r\in V^r$, a Virtual Node (VN)\\
\hline
~~~~~~~~$e^r$ & $e^r \in E^r$, a Virtual Link (VL)\\ 
\hline
~~~~~~$v^r_iv^r_j$ & $v^r_iv^r_j \in E^r$, the VL connecting $v^r_i \in V^r$ and $v^r_j \in V^r$\\ 
\hline
$P^r(v^r_1v^r_2...v^r_n)$ &A path VNR with $n$ VNs, $\textit{i.e.}$, from $v^r_1$, through $v^r_{j, 2\leq j \leq n-1}$, to $v^r_n$. \\
\hline
~~~~~~$|P^r|$  &The length of $P^r$, $\textit{i.e.}$, the number of VLs in $P^r$ \\
\hline
$C^r(v^r_1,...,v^r_nv^r_1)$ & A cycle VNR with clockwise order of VNs, where $n$ is the number of VNs on $C^r$\\
\hline
\hline
~~~~$CPU(v^s)$ &The CPU capacity of SN $v^s$ \\
\hline
~~~~$BW(e^s)$ & \\ or $BW(v^s_iv^s_j)$ & The bandwidth capacity of SL $e^s$ or $v^s_iv^s_j$\\
\hline
~~~~$CPU(v^r)$ & The CPU demand of VN $v^r$ \\
\hline
~~~~$BW(e^r)$ & \\ or $BW(v^r_iv^r_j)$ & The bandwidth demand of SL $e^r$ or $v^r_iv^r_j$\\
\hline
~~~~~~$\mathcal{F}_{v^r}$ & $\{v^s \in V^s | CPU(v^s) \geq CPU(v^r) \}$, $\textit{i.e.}$, the set of feasible SNs whose CPUs are not smaller than VN $v^r$'s\\
 \hline
~~$\mathcal{F}_{e^r}$ or $\mathcal{F}_{v^r_iv^r_j}$ & $\{e^s \in E^s| BW(e^s) \geq BW(e^r) \text{~or~} BW(v^r_iv^r_j) \}$, \textit{i.e.}, the set of feasible SLs whose bandwidths are not smaller than $e^r$ or $v^r_iv^r_j$\\
\hline
~~~$v^r \rightarrow v^s$ & $v^r$ is embedded on $v^s$\\
\hline
~~~$e^r \rightarrow e^s$ & The embedded substrate path of $e^r$ which passes through $e^s$\\
\hline
\end{tabular}
\end{table}

\subsection{Network Models}
\label{subsec: NM}
\subsubsection{Substrate Network} In this paper, the substrate network $G^s(V^s,E^s)$ is an undirected connected graph. Usually, there are two main resources on the substrate network managed by the InP, the computing capabilities of the SNs and the bandwidths of the SLs. Here, we denote by $CPU(v^s)$ the computing capability of each SN $v^s$, and denote $BW(e^s)$ as the bandwidth of SL $e^s$. Figure \ref{fig:f3} gives an example of a 4-node substrate network, where each number in the square indicates the $CPU$ capability of the SN and each number beside the SL indicates its $BW$.    
\begin{figure}[!htb]
\centering
\includegraphics[height=0.23\textwidth]{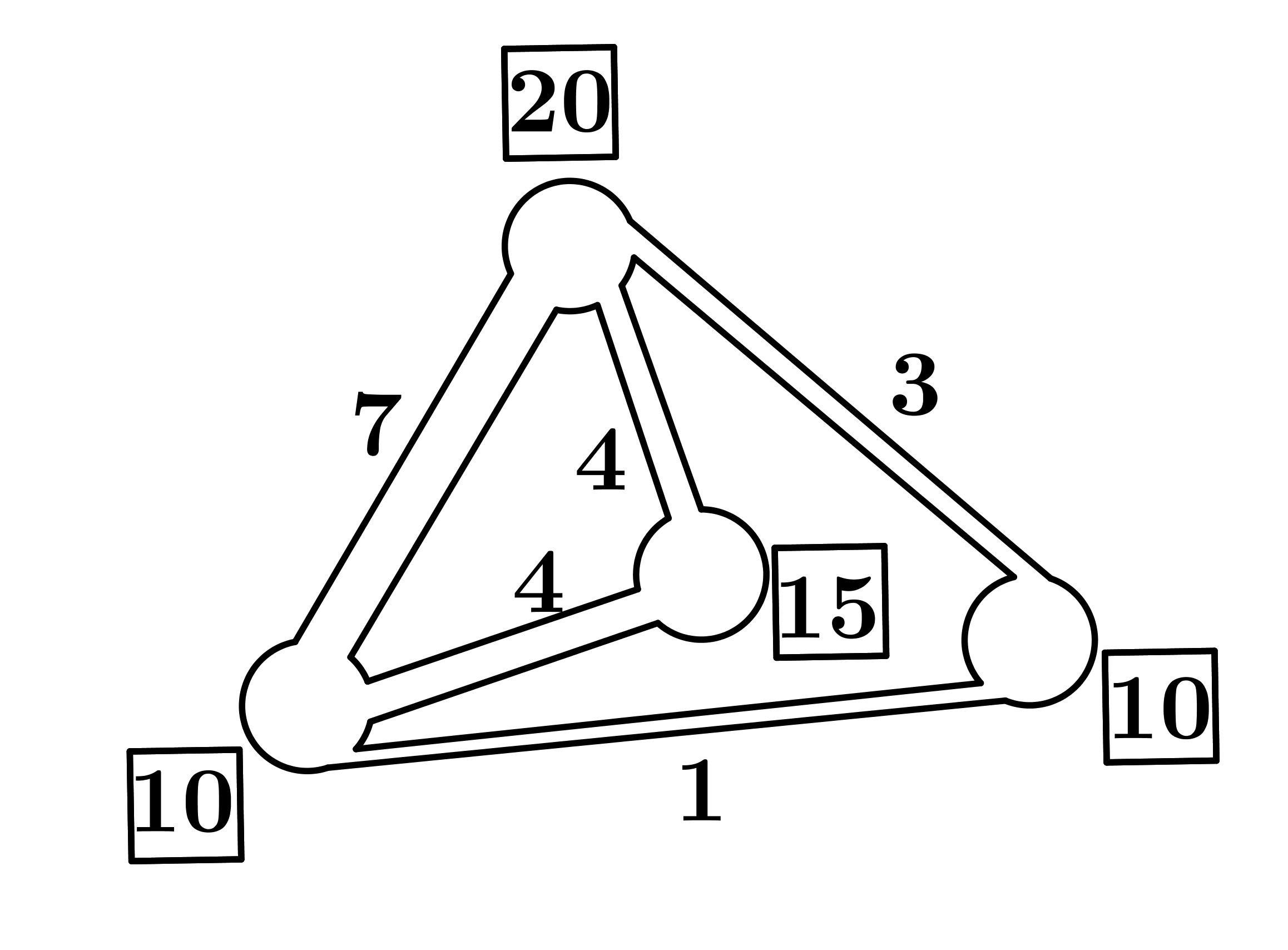}
\caption{A 4-node substrate network.}
\label{fig:f3}
\end{figure}

\subsubsection{Virtual Network Request} A VNR $G^r(V^r,E^r)$ is still modeled by an undirected connected graph, which is constructed by the SP according to the service demands of clients. The demanded computing capability of each $v^r \in V^r$ is $CPU(v^r)$. 
The demanded bandwidth of each $e^r \in E^r$ is $BW(e^r)$. In this work, we focus on two special topologies of VNRs $\textit{i.e.}$, paths and cycles. Figure \ref{fig:f4} illustrates a path VNR and a cycle VNR, where, similarly, the numbers in squares (beside VLs respectively) indicate the corresponding $CPU$s of VNs ($BW$s of VLs respectively).
\begin{figure}[!htb]
\centering
\includegraphics[height=0.15\textwidth]{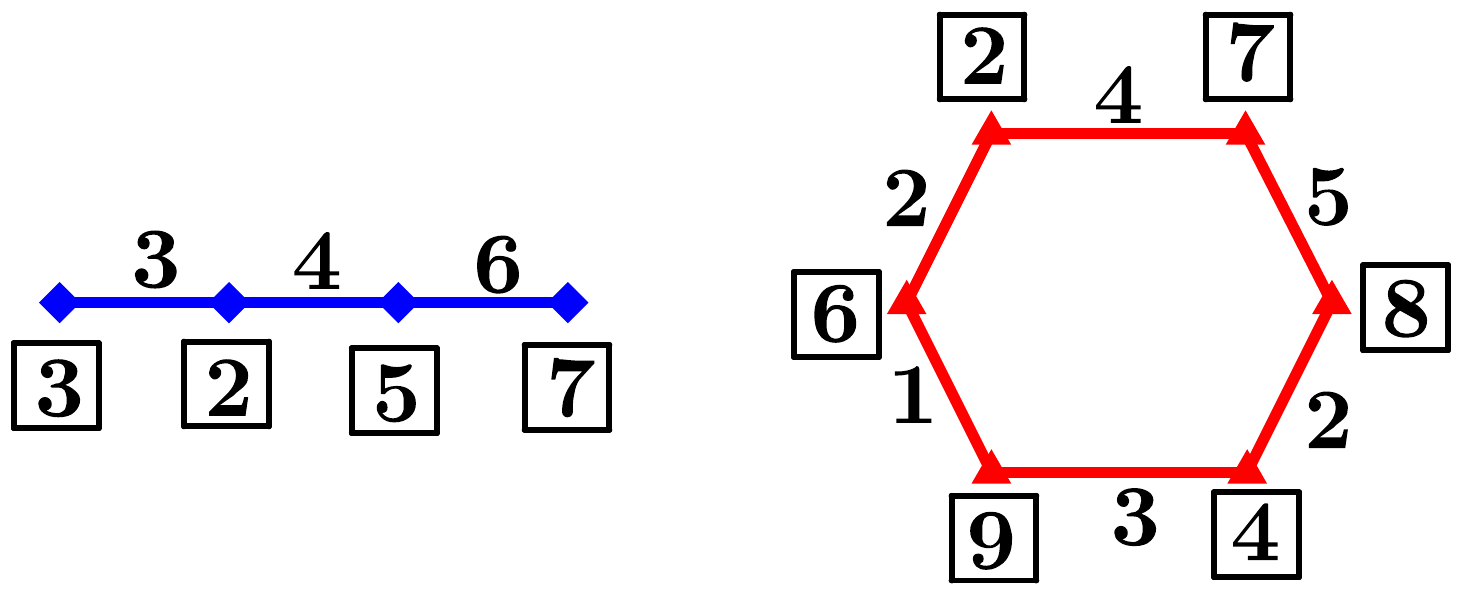}
\caption{An example of path and cycle VNRs.}
\label{fig:f4}
\end{figure}


\subsection{Problem Formulation}
\label{subsec:pf}

\subsubsection{Constraints}
As mentioned above, the VNE problem contains two constraints, $\textit{Node Mapping}$ and $\textit{Link Mapping}$ constraints.

\begin{figure}[!htp]
\centering
\includegraphics[height=0.3\textwidth]{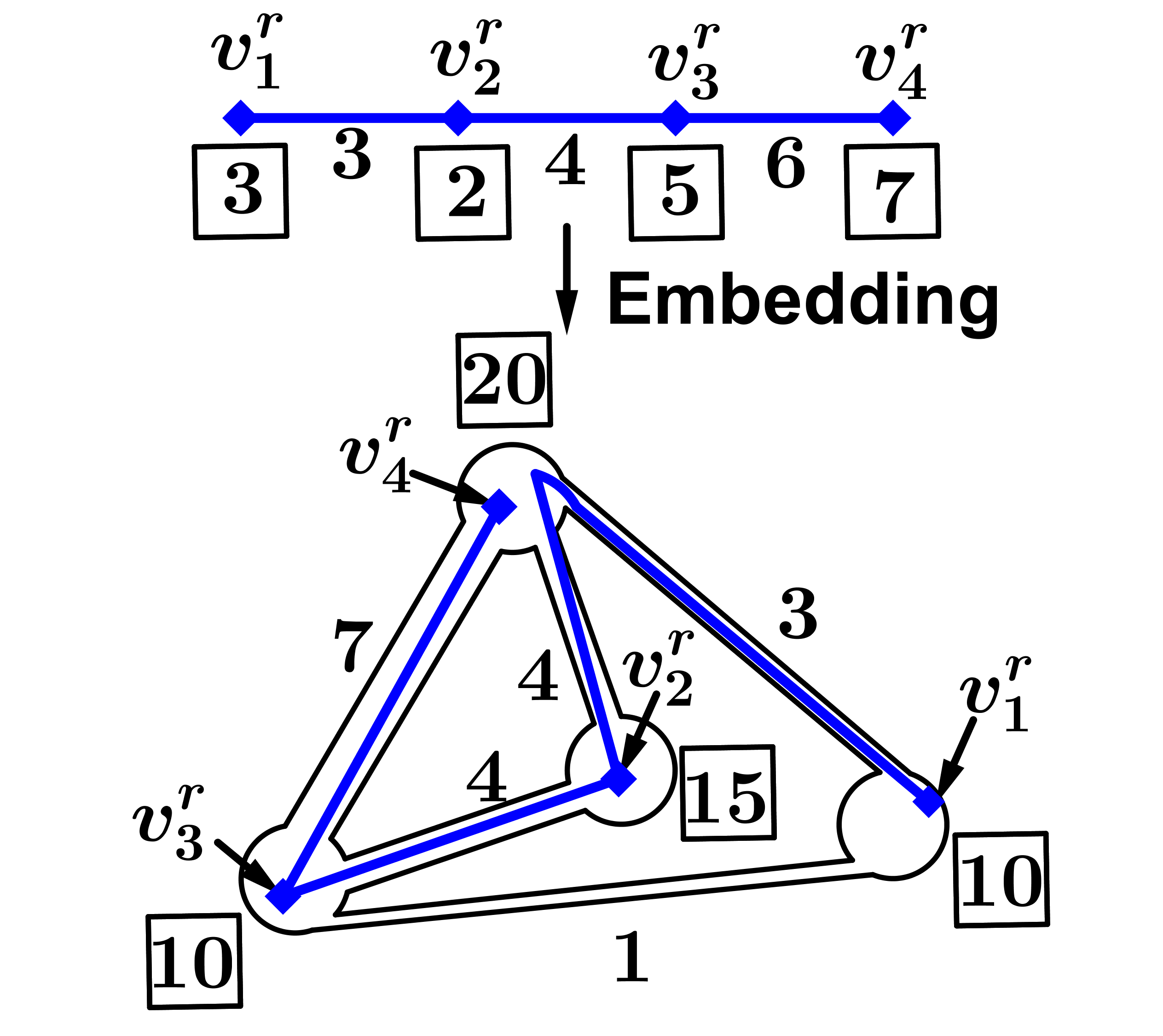}
\caption{Node and link mapping.}
\label{fig:f5}
\end{figure}

\paragraph{$\textit{Node Mapping Constraint}$}
For the node mapping,  given the $G^s(V^s,E^s)$ and a VNR $G^r(V^r,E^r)$, each VN $v^r \in V^r$ must be embedded onto such an SN $v^s \in V^s$, $\textit{i.e.}$, $v^r \rightarrow v^s$, that $CPU(v^s) \geq CPU(v^r)$. Meanwhile, for each SN $v^s$ on which some VNs $v^r \in V^r$ are embedded, $CPU(v^s) \geq \sum_{v^r\rightarrow v^s} CPU(v^r)$. Figure \ref{fig:f5} demonstrates the node mapping where the numbers in squares indicate the $CPU$ requirements or capabilities of VNs or SNs respectively. Besides, any two VNs $v^r_i$ and $v^r_j$ of a same VNR can not be embedded on a same SN $v^s$, \textit{i.e.}, if $v^r_1 \rightarrow v^s$ and $v^r_2\rightarrow v^s$ then $v^r_1=v^r_2$ \cite{s1:b8}.


\paragraph{$\textit{Link Mapping Constraint}$} For each VL $v^r_iv^r_j$, assuming $v^r_i \rightarrow v^s_i$ and $v^r_j \rightarrow v^s_j$, $v^r_iv^r_j$ should be embedded on a substrate path $P_{v^s_iv^s_j} \in \mathcal{P}_{v^s_iv^s_j}$, and for each SL $e^s$ on $P_{v^s_iv^s_j}$, $BW(e^s) \geq BW(v^r_iv^r_j)$. 
Meanwhile, for each SL $e^s$, through which the embedded substrate paths of VLs $e^r$ pass, \textit{i.e.}, $e^r \rightarrow e^s$, $BW(e^s) \geq \sum_{e^r \rightarrow e^s}BW(e^r)$ as shown in Fig \ref{fig:f5}.

Here, we give the definition of the decision problem \textbf{$Emb(\cdot,\cdot)$} as follows.

\begin{definition}
\label{de1}
$Emb(G^s,G^r)$ is such a decision problem that its answer is \textbf{Yes} iff the VNR $G^r$ can be embedded on the substrate network $G^s$ satisfying the node and link mapping constraints, and \textbf{No} otherwise.  
\end{definition}

The definition can be extended to a set of $n$ VNRs: $Emb(G^s,\{G^r_1,G^r_2,...,G^r_n\})$ whose answer is \textbf{Yes} iff all of $G^r_{i,1\leq i \leq n}(V^r_i,E^r_i)$ can be embedded on the $G^s$ while simultaneously satisfying the node and link mapping constraints, and \textbf{No} otherwise. Here, the node mapping constraint of a set of VNRs embedded on a substrate network is that for each SN $v^s$, $CPU(v^s) \geq \sum_{v^r \rightarrow v^s} CPU(v^r), \forall v^r \in V^r_1 \cup V^r_2\cup...\cup V^r_n$, and the link constraint means for each SL $e^s$, $BW(e^s)\geq \sum_{e^r \rightarrow e^s} BW(e^r)$, $\forall e^r\in E^r_1 \cup E^r_2 \cup...\cup E^r_n$.

\subsubsection{Objective Functions}
There are two main objective functions in the study of the VNE problem: The acceptance ratio and the revenue.
\paragraph{The Acceptance Ratio (AcR)}
Given a substrate network $G^s$ and a set $\{G^r_1,G^r_2,...,G^r_n\}$ of VNRs, the objective  is to maximize the number of VNRs that can be embedded on the $G^s$. We denote by \textbf{AcR} the acceptance ratio problem as formally defined below.

\begin{equation*}
\begin{aligned}
&\text{Maximize} \quad |S| \textrm{~~~~(\textbf{AcR})},\\
& \textbf{\textit{s.t.}} \quad Emb(G^s,S) ~ \text{is \textbf{Yes}}, \\ & \quad ~~~S \subseteq \{G^r_1,G^r_2,...,G^r_n\}. 
\end{aligned}
\end{equation*}

\paragraph{The Revenue (Rev)} Each VNR $G^r_i$ is associated with a revenue $w_i$. Its objective is to maximize the total revenue of VNRs that can be embedded on $G^s$. We denote by \textbf{Rev} the revenue problem defined below.   
\begin{equation*}
\begin{aligned}
&\text{Maximize} \quad \sum_{G^r_i\in S}w_i \textrm{~~~~(\textbf{Rev})},\\
& \textbf{\textit{s.t.}} \quad Emb(G^s,S) ~ \text{is \textbf{Yes}}, \\ & \quad ~~~S \subseteq \{G^r_1,G^r_2,...,G^r_n\}. 
\end{aligned}
\end{equation*}

One may notice that AcR actually is a special case of Rev by setting each revenue $w_i$ to be one. In this paper, the set of VNRs $\{G^r_1,G^r_2,...,G^r_n\}$ is particularly a path set or a cycle set, \textit{i.e.}, $\{P^r_1,P^r_2,...,P^r_n\}$ or $\{C^r_1,C^r_2,...,C^r_n\}$.

\section{Theoretical Hardness Analysis of Path Embedding}
\label{sec:pm}
In this section, we explore the theoretical harness of the path embedding problem. 
Given a substrate network $G^s$ and a path VNR $P^r$, an elementary and essential question is how hard $Emb(G^s,P^r)$ is. Before answering it, we give some terminologies in graph theory as follows.

\begin{itemize}
\item \textit{Trail}: In a graph $G(V,E)$, a \textit{trail} is such a subgraph that can be expressed as a sequence of vertices "$v_1,v_2,...,v_n$", where for any two adjacent vertices $v_i$ and $v_{i+1}$, $1\leq i \leq n-1$, $v_iv_{i+1}$ is an edge, \textit{i.e.}, $v_iv_{i+1}\in E$ and no repeated edge occurs in the trail, \textit{i.e.}, the pair ($v_i$, $v_{i+1}$) only occurs once in the trail. A \textit{closed trail} is such a trail "$v_1,v_2,...,v_n$", with $v_1=v_n$. Figure \ref{fig:f7} illustrates a trail $T_1$ (left in blue), a closed trail $T_2$ (right in red), and the corresponding vertex sequences.

\begin{figure}[!htb]
\centering
\includegraphics[height=0.2\textwidth]{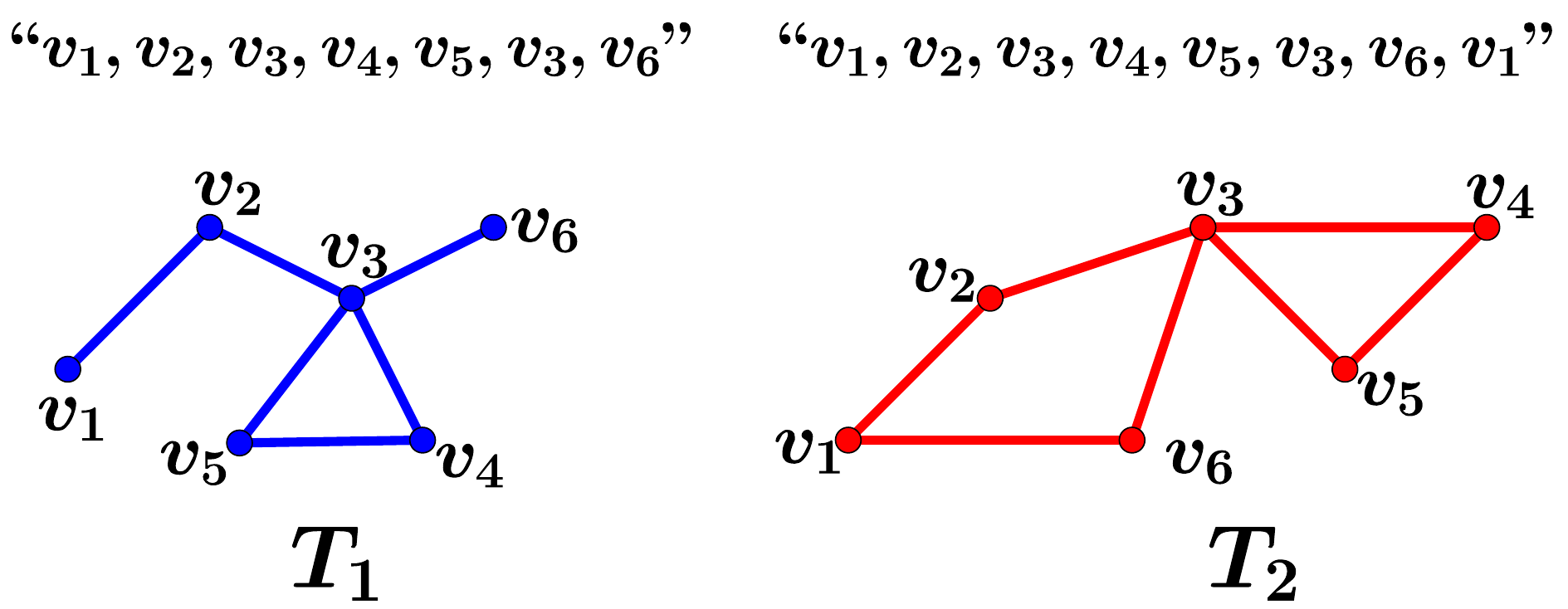}
\caption{An example of trail and closed trail.}
\label{fig:f7}
\end{figure}

\item \textit{Eulerian trail}: In a graph $G$, an Eulerian trail is a trail traversing all edges. Notice that not every graph has an Eulerian trail. An \textit{Eulerian circuit} is a closed Eulerian trail. Therefore, in Fig. \ref{fig:f7}, if we treat the two trails as two graphs (not subgraphs), then they can also be regarded as an Eulerian trail and an Eulerian circuit respectively. 

\item \textit{Eulerian graph}: A graph $G$ is called Eulerian graph iff it has an Eulerian circuit. In Fig. \ref{fig:f7}, if we treat $T_2$  as a graph, then $T_2$ is an Eulerian graph.

\item \textit{Supereulerian graph}: A graph $G$ is called Supereulerian graph iff it contains a spanning subgraph\footnote{A spanning subgraph of $G$ is such a subgraph that contains all vertices of $G$.} which is an Eulerian graph. Similarly, if a graph $G$ is formed by adding one edge $v_2v_6$ to $T_2$ in Fig. \ref{fig:f7}, \textit{i.e.}, $G=T_2+v_2v_6$, then $G$ is a Supereulerian graph.
\end{itemize}

\subsection{The Hardness of $Emb(G^s,P^r)$}
\label{subsec:hard}

Let $G^s(V^s,E^s)$ and $P^r(V^r, E^r)$ be the substrate network and the path VNR respectively. To characterize the hardness of $Emb(G^s,P^r)$, we adopt the \textit{uniform setting} of the $CPUs$ and $BWs$. 
\begin{definition}
\textit{uniform setting}: \circled{1} $\forall v^s \in V^s$, $CPU(v^s)=2$ and  $\forall e^s \in E^s$, $BW(e^s)=1$;  \circled{2} $\forall v^r \in V^r$, $CPU(v^r)=1$ and  $\forall e^r \in E^r$, $BW(e^r)=1$.
\end{definition}
The number of VNs that can be embedded is at most $|V^s|$. Thus, a special case of $Emb(G^s,P^r)$ is to answer whether a path VNR $P^r(v^r_1v^r_2,...,v^r_{|V^s|})$ can be embedded onto $G^s(V^s,E^s)$. 
The following lemma gives the necessary and sufficient condition that whether a $P^r$ of $|V^s|$ VNs can be embedded onto $G^s$. 
\begin{lemma}
\label{lem0}
In the uniform setting, a $P^r$ of $|V^s|$ VNs can be embedded on the $G^s$ \textbf{if and only if} $G^s$ has a trail traversing all SNs. In other words, $G^s$ contains a spanning subgraph having an Eulerian trail.
\end{lemma}

\begin{proof}
If $P^r(v^r_1v^r_2,...,v^r_{|V^s|})$ can be embedded on $G^s$ by, \textit{w.l.o.g.}, $v^r_i \rightarrow v^s_i, \forall i$, then we can express this embedding way as a sequence "$v^s_1P_{v^s_1v^s_2}v^s_2P_{v^s_1v^s_3}...v^s_{|V^s|}$". This sequence, if expressing each $P_{v^s_iv^s_{i+1}}$ by its vertex sequence, in fact is a trail which traverses all SNs in $G^s$ (but not necessary all SLs). 

Conversely, if $G^s$ has a trail, say $T$, traversing all SNs, let, \textit{w.l.o.g.}, "$v^s_1,...,v^s_2,...,v^s_i,...,v^s_{|V^s|}$" be $T$'s corresponding vertex sequence. By the definition of trail, there must be a path connecting $v^s_i$ and $v^s_{i+1}$, say $P_{v^s_iv^s_{i+1}}, \forall i$, and $P_{v^s_iv^s_{i+1}} \cap P_{v^s_jv^s_{j+1}}=\emptyset, \forall i \neq j$. Therefore we can embed $P^r$ by "$v^s_1P_{v^s_1v^s_2}v^s_2P_{v^s_1v^s_3}...v^s_{|V^s|}$".
\end{proof}

\begin{lemma}
\label{lem1}
The hardness of determining whether a graph contains a spanning subgraph having an Eulerian trail (or $SSET$ in short) is equivalent\footnote{The equivalence is under the polynomial-time Turing reduction.} to that of determining whether a graph is a Supereulerian graph (or $SG$). 
\end{lemma}

Consequently, by Lemmas \ref{lem0} and \ref{lem1}, we have Theorem \ref{the1}.
\begin{theorem}
\label{the1}
In the uniform setting, the hardness of $Emb(G^s,P^r)$ problem is not less than that of determining whether $G^s$ is a Supereulerian graph.
\end{theorem}

Unfortunately, given a graph $G$, it is $\mathcal{NP}$-hard to determine whether $G$ is a Supereulerian graph \cite{s4:b1}. Thus, $Emb(G^s,P^r)$ is also $\mathcal{NP}$-hard.

\subsection{The inapproximability of path embedding}

After obtaining the $\mathcal{NP}$-hardness of $(G^s, P^r)$, immediately, we can get the inapproximability of the AcR and Rev problems of path embedding in the uniform setting as follows.

\begin{theorem}
\label{the:inapr1}
For path embedding in the uniform setting, both the AcR and Rev problems have an $\mathcal{NP}$-hard gap $[\epsilon, 1]\footnote{If an optimization problem has an $\mathcal{NP}$-hard gap $[\alpha,\beta]$, then it has no polynomial-time algorithm  of an approximation ratio $\frac{\beta}{\alpha}$ unless $\mathcal{NP}=\mathcal{P}$ \cite{sr:b2}.}, \forall 0<\epsilon <1$, \textit{i.e.}, unless $\mathcal{NP}=\mathcal{P}$, there is no polynomial-time algorithm of an approximation ratio $\frac{1}{\epsilon}$ for both problems.
\end{theorem}  

\begin{proof}
Given an instance of the AcR problem consisting of a $G^s$ and a $P^r$, if the answer of $Emb(G^s,P^r)$ is \textbf{Yes}, then $OPT_{AcR}\geq 1$, otherwise, $OPT_{AcR} < \epsilon$, where $OPT_{AcR}$ is the optimal solution of the AcR. Since $Emb(G^s,P^r)$ is $\mathcal{NP}$-hard, the AcR problem in the uniform setting has an $\mathcal{NP}$-hard gap\footnote{For a maximization problem $\Pi$ and an $\mathcal{NP}$-hard decision problem $\Lambda$, if $\Lambda$ is \textbf{Yes} $\Rightarrow$ $OPT_{\Pi} \geq \beta$ and $\Lambda$ is \textbf{No} $\Rightarrow$ $OPT_{\Pi} < \alpha$, where $OPT_{\Pi}$ is the optimal solution of $\Pi$, then $\Pi$ has an $\mathcal{NP}$-hard gap $[\alpha,\beta]$ \cite{sr:b2}.} $[\epsilon, 1]$. For the Rev, as the AcR is its special case, the proof also follows.
\end{proof}

In reality that $CPU$ and $BW$ are arbitrary, the path embedding should be more difficult than in the uniform setting. Here, we present the following theorem to explicitly show the inapproximability of both the AcR and Rev for path embedding in the realistic settings.

\begin{theorem}
\label{the:srin}
For path embedding in reality, unless $\mathcal{NP} \subseteq ZPTIME(n^{polylog(n)})\footnote{$ZPTIME(n^{polylog(n)})$ is the set of languages that have randomized algorithms that always give the correct answer and have expected running time $n^{polylog(n)}$.}$, there is no $\mathcal{O}(\log^{\frac{1}{2}-\epsilon}|E^s|)$ approximation for both the AcR and Rev problems,  where $|E^s|$ is the number of SLs in the substrate network.
\end{theorem}

\begin{proof}
First, we introduce the Edge-Disjoint Paths (EDP) problem: Given a connected graph $G(V,E)$ and a set of pairs of vertices $\{(s_1, t_1), (s_2,t_2),...,(s_k,t_k)\}$ on it, the objective of the EDP is to connect as many pairs as possible via edge-disjoint paths. Unless $\mathcal{NP} \subseteq ZPTIME(n^{polylog(n)})$, there is no $\mathcal{O}(\log^{\frac{1}{2}-\epsilon}|E|)$ approximation for the EDP problem \cite{sr:b4}.

Now given an instance of the EDP, \textit{i.e.}, a graph $G(V,E)$ where $V=\{v_1,...v_i,...,v_n\}$ and a set of pairs of vertices $\{(s_1, t_1), (s_2,t_2),...,(s_k,t_k)\}$, we translate it into an instance of the AcR in path embedding as follows. 
To complete the translation, for each $v_i \in V$, we first denote by $N_{v_i}$ the total number of $v_i$ appearing in the sets $\{s_i\}^k_{i=1}$ and $\{t_i\}^k_{i=1}$.

 \begin{figure}[!htb]
 \centering
 \includegraphics[height=0.2\textwidth]{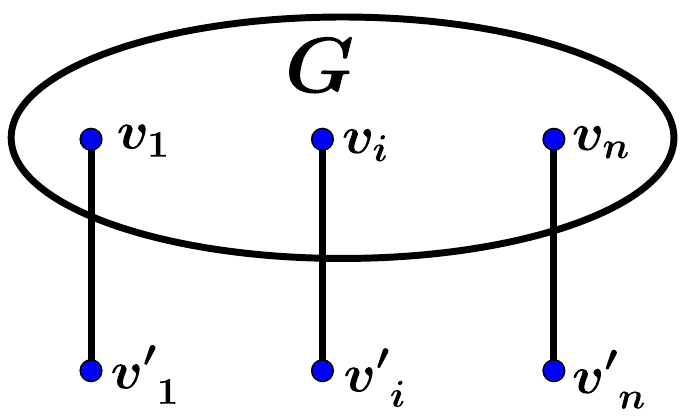}
 \caption{The constructed substrate network.}
 \label{fig:sr17}
 \end{figure}

The construction of the substrate network $G^s$ is that: The set of SNs $V^s$ consists of two parts, $V=\{v_1,...v_i,...,v_n\}$ and $V'=\{v'_1,...,v'_i,...,v'_n\}$ where $V'$ is a copy of $V$. The set of SLs $E^s$ also consists of two parts, $E$ and $\{v_iv'_i\}^n_{i=1}$ as shown in Fig. \ref{fig:sr17}.

The setting of $CPUs$ and $BWs$ is as follows.
\begin{itemize}
\item For each $v_i \in V$, $CPU(v_i)=1$ and for each $e \in E$, $BW(e)=1$.
\item For $V'$, $CPU(v'_1)=C_{v_1} \times N_{v_1}$ where $C_{v_1}=2$ and for $2 \leq i \leq n$, $CPU(v'_i)=C_{v_i}\times N_{v_i}$ where $C_{v_i}=CPU(v'_{i-1})+1$.
\item For $\{v_iv'_i\}^n_{i=1}$, $BW(v_nv'_n)=B_{v_n} \times N_{v_n}$ where $B_{v_n}=2$ and for $n-1 \geq i \geq 1$, $BW(v_iv'_i)=B_{v_i} \times N_{v_i}$ where $B_{v_i}=BW(v_{i+1}v'_{i+1})+1$. 
\end{itemize}

 \begin{figure}[!htb]
 \centering
 \includegraphics[height=0.13\textwidth]{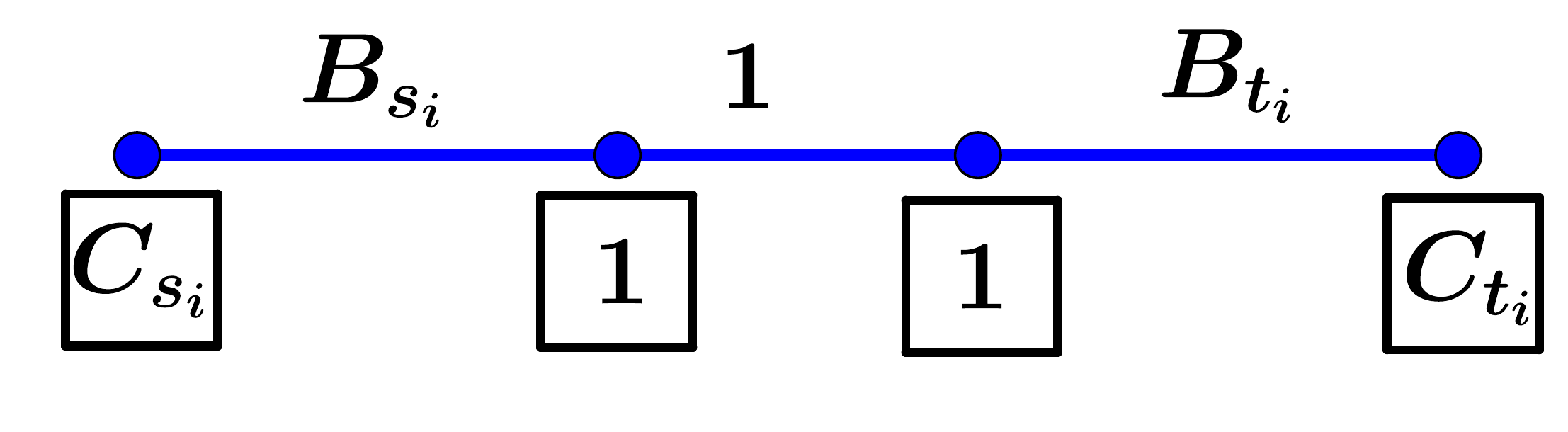}
 \caption{The $i$-th path VNR $P^r_i$.}
 \label{fig:sr18}
 \end{figure}

The path VNRs set is $\{P^r_1,...,P^r_i,...,P^r_k\}$, where the $i$-th $P^r_i$, which corresponds to $(s_i,t_i)$, consists of 4 VNs and its $CPUs$ and $BWs$ requirements are shown as in Fig. \ref{fig:sr18}

From the setting of $CPUs$ and $BWs$, we can see that $C_{v_i}$ is increasing with the $i$ growing while $B_{v_i}$ is decreasing. Thus, following the constraints of node and link mappings, it is easy to see that the two ends of the $P^r_i$ must be embedded on the copies of $s_i$ and $t_i$ respectively, if the $P^r_i$ is embedded.

Hence, for any pairs we can connect by edge-disjoint paths in the instance of the EDP, we can embed their corresponding path VNRs in the substrate network, vice versa.

Besides, since the number of SLs $|E^s|=|E|+|V| \leq 3|E|$, we get the inapproximability.
\end{proof}

\section{Practical Algorithm Design for Path Embedding}
\label{sec:rm}

Theorem \ref{the:srin} implies that for path embedding, it is implausible to find a polynomial-time algorithms of a proper approximation ratio. Thus, we should turn our attention from developing approximation algorithms to designing better heuristic algorithms, which is able to capture the "nature" of the path VNE problem. The following lemma captures the "nature" to some extent. 

\begin{lemma}
\label{exmp}
In the uniform setting, if the substrate network is a path denoted by $P^s$, given a set of $\{P^r_1,P^r_2,...,P^r_n\}$, then the AcR and Rev problems can be solved in polynomial time. 
\end{lemma}

\begin{proof}
\begin{figure}[!htb]
\centering
\includegraphics[height=0.3\textwidth]{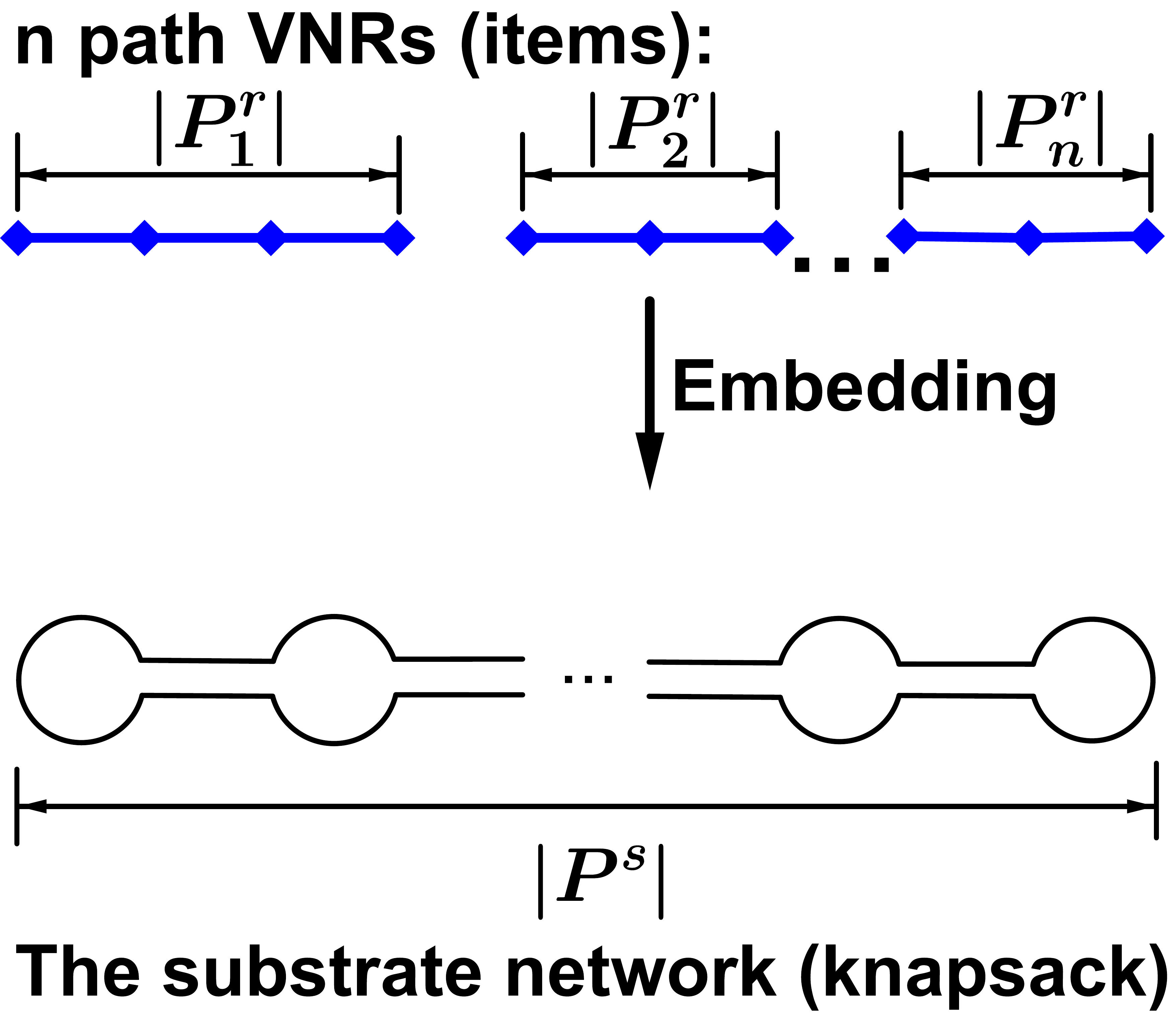}
\caption{An example of the substrate network being a path.}
\label{fig:f10}
\end{figure}
We use Fig. \ref{fig:f10} to demonstrate the substrate path $P^s$ and the set of path VNRs  $\{P^r_1,P^r_2,...,P^r_n\}$. In this special case, we can regard $P^s$ as a "knapsack" with a capacity of $|P^s|$ (edge number), and the $n$ path VNRs $P^r_{i, 1\leq i \leq n}$ as $n$ items with a size of $|P^r_{i, 1\leq i \leq n}|$ respectively.   

The AcR can be easily solved in this way: First arrange the $n$ path VNRs in the increasing order of their sizes, then sequentially pack them into the "knapsack" until cannot (This embedding is feasible since each middle SN $v^s$ of $CPU(v^s)=deg(v^s)=2$, and at most two VNs will be embedded on $v^s$).

The Rev problem in this special case is equivalent to the classical 0-1 Knapsack Problem (KP) and it thus can be solved by a dynamic programming algorithm of time complexity $\mathcal{O}(n|P^s|)$ \cite{s4:b4} \footnote{Although KP is $\mathcal{NP}$-hard, for the substrate network, whose space complexity is $|P^s|$ (not $\log(|P^s|$)), the dynamic programming algorithm runs in polynomial time (rather than pseudo-polynomial time).}.
\end{proof}

Thus, in the uniform setting, if $G^s$ is a path, then the AcR and Rev problems can be easily solved by leveraging KP. This result relatively reflects some essentials of the path VNE problem that can be regarded as "packing" (embedding) a set of "items" (VNRs) into a special "knapsack" (the substrate network). Inspired by this, given a substrate network $G^s$ and a set of  $\{P^r_1,P^r_2,...,P^r_n\}$, we propose a framework of algorithm design for the realistic settings. The main idea is described as follows. First, we decompose $G^s$ into several substrate paths. This phase is thus called path decomposition. By regarding each substrate path as a knapsack and each $P^r_j$ as an item with size $|P^r_j|$ and profit $w_j$, and we then pack these items into  multiple knapsacks, which can be formulated as a Multiple Knapsack problem (MKP). Finally, we assign the $CPU$ and $BW$ resources to those packed path VNRs, and it corresponds to the Multi-Dimensional Knapsack Problem (MDKP). To this end, we review the two well-studied MKP and MDKP. 

\textbf{Multiple Knapsack Problem (MKP)} 

MKP\cite{s4:b4} is a classical variation of KP. In MKP, there are a set of knapsacks $M:=\{1,...,i,...,m\}$ each with positive capacities $b_i$, and a set of items $N:=\{1,...j,...,n\}$ each with size $s_j \geq 0$ and profit $w_j \geq 0$. The goal is to find a subset of the $n$ items of maximum profit which can be packed into the $m$ knapsacks. 



In this paper, we measure the time complexity of solving an instance of MKP by its numbers of knapsacks and items, $m$ and $n$ respectively, denoted by $T_{MKP}(m,n)$.

\textbf{Multi-Dimensional Knapsack Problem (MDKP)}

MDKP\cite{s4:b4} is another well-known variation of KP. In $d$ dimensional MDKP denoted by $d$-DKP, there are a knapsack of $d$-dimensional positive capacity attributes $(b_1,..,b_i,...,b_d)$ and a set of items $N:=\{1,...j,...,n\}$ each with profit $w_j$ and $d$-dimensional size attributes $(s_{j1},...s_{ji},...,s_{jd})$, where all of $b_i$ and $w_j$ and $s_{ji}$ are non-negative. The goal is to find a subset of the $n$ items of maximum profit which can be packed into the knapsack while not exceeding each of $d$-dimensional capacity attributes. 




In this paper, we measure the time complexity of solving an instance of MDKP by its numbers of dimensions $d$ and items $n$ respectively, denoted by $T_{MDKP}(d,n)$. Next, we assume that the substrate network $G^s(V^s,E^s)$ and a set of path VNRs $\{P^r_1,P^r_2,...,P^r_n\}$ are given as the input of our algorithm. 

\subsection{Path Decomposition Phase}
In this phase, we decompose the substrate network into a set of substrate paths. These decomposed paths are treated as the multiple knapsacks for the optimization of the next phase. Intuitively, if these "knapsacks" are of bigger capacities, \textit{i.e.}, much longer, the embedding optimization in the next phase by the MKP will be better. 
We extract a substrate path $P^s$ by finding the longest path in a Depth-first Search Tree (DST) of $G^s$. We repeat the process by keeping extracting substrate paths until $G^s$ is completely decomposed into a set of $P^s_i$, $1\leq i \leq m$, where $m$ is the number of substrate paths obtained.

\subsection{Embedding by MKP}
After the path decomposition phase, we regard each path VNR $P^r_j$ as an item with size $|P^r_j|$ and profit $w_j$, and treat each substrate path $P^s_i$ as a knapsack with capacity $|P^s_i|$ as shown in Fig. \ref{fig: mkpembed}. Obviously, it is an $m$-knapsacks-$n$-items MKP, where $m \leq |E^s|$.

 \begin{figure}[!htb]
 \centering
 \includegraphics[width=0.5\columnwidth]{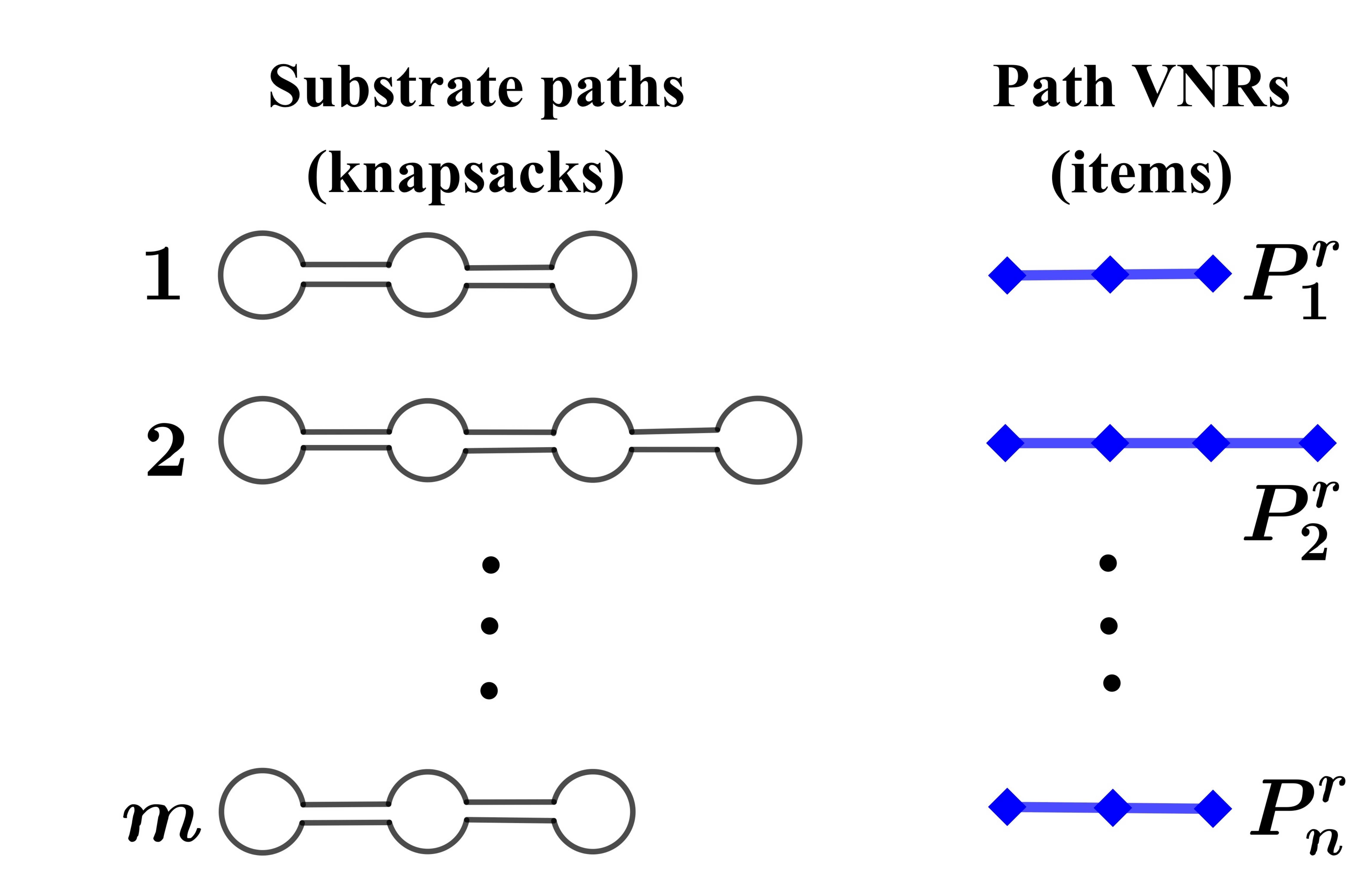}
 \caption{MKP embedding.}
 \label{fig: mkpembed}
 \end{figure}
 \begin{figure}[!htp]
 \centering
 \includegraphics[width=0.5\columnwidth]{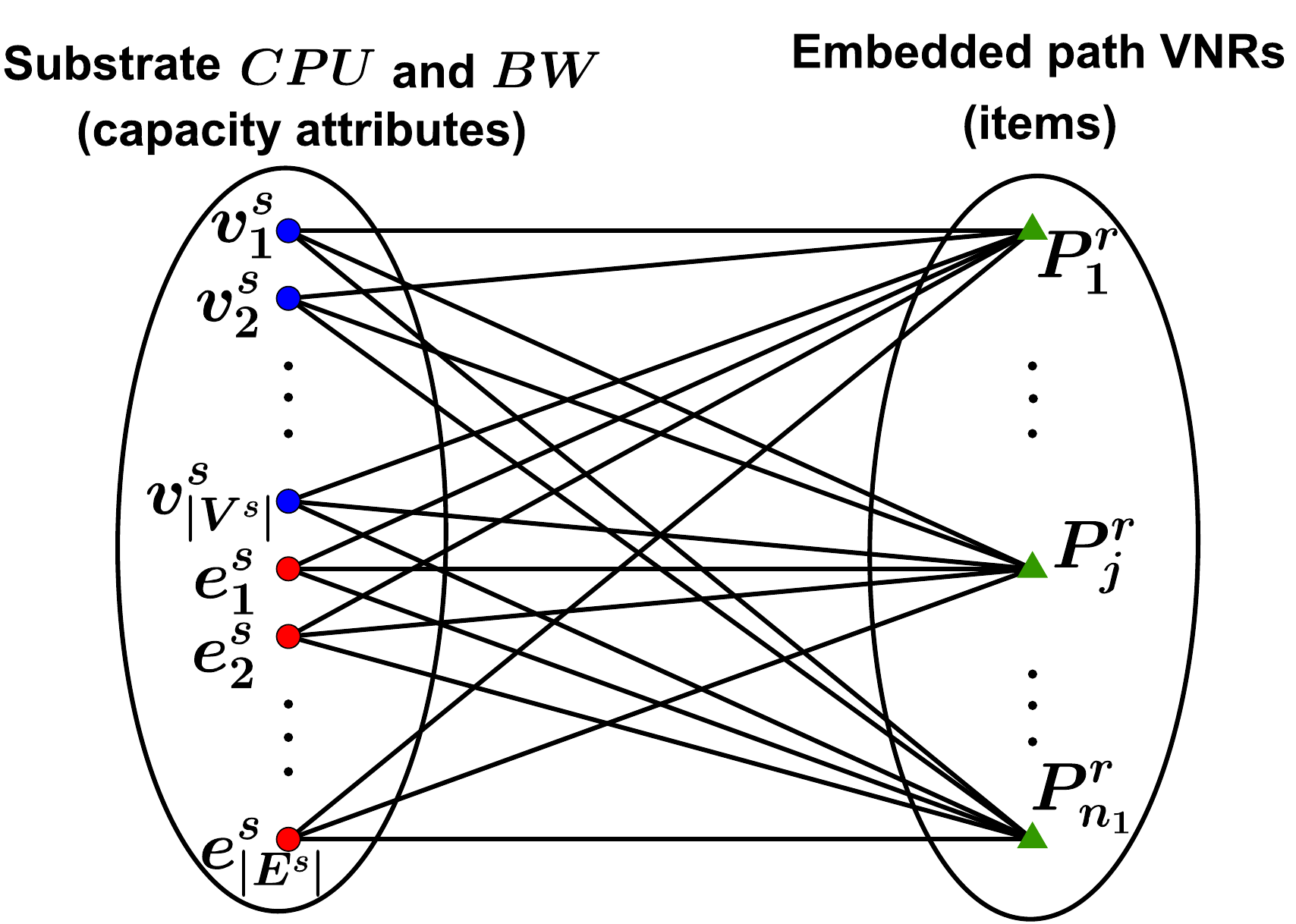}
 \caption{MDKP assignment.}
 \label{fig: mdkpassign}
 \end{figure}

\subsection{Resource Assignment by MDKP}
After embedding some path VNRs by MKP without considering $CPU$s and $BW$s, we need to assign the corresponding demanded $CPU$ and $BW$ resources to as many embedded path VNRs as possible. Here, we treat each embedded $P^r_j$ as an item with $(|V^s|+|E^s|)$-dimensional size attributes: $(s_{jv^s_1},s_{jv^s_2},...,s_{jv^s_{|V^s|}},s_{je^s_1},s_{je^s_2},...,s_{je^s_{|E^s|}})$. For the first $|V^s|$ attributes, if some VNs $v^r$ of the $P^r_j$ are embedded on an SN, say $v^s_k$, then the attribute $s_{jv^s_k}=\sum_{v^r \rightarrow v^s_k}CPU(v^r)$, otherwise 0. For the last $|E^s|$ attributes, if a VL $e^r$ of the $P^r_j$ is embedded on an SL, say $e^s_l$, then the attribute $s_{je^s_l}=BW(e^r)$, otherwise 0. Finally, the array of capacity attributes of the knapsack is that: $(CPU(v^s_1),...,CPU(v^s_{|V^s|}),BW(e^s_1),...,BW(e^s_{|E^s|})$. How to assign resources to these embedded path VNRs to maximize revenue is obviously a $(|V^s|+|E^s|)$-DKP with $n_1$ items as shown in Fig. \ref{fig: mdkpassign}, where $n_1 \leq n$ is the number of embedded path VNRs by MKP.


\subsection{Final Assembled Algorithm and Time Complexity}

After the resource assignment, we can update the $CPU$ and $BW$ of each SN and SL, resulting in a remained substrate network. We then continue the whole process, from path decomposition to resource assignment, to embed the rest path VNRs until no more paths can be embedded. The final assembled algorithm is shown in Algorithm \ref{pe}.
\begin{algorithm}[!h]
    \caption{Procedure of Path-Embedding (PE)}
    \label{pe}
		\SetKwInOut{Input}{Input}
		\SetKwInOut{Output}{Output}
		\SetKw{KwAnd}{and}
		\SetKw{KwSuch}{s.t.}
		\Input{A substrate network $G^s(V^s,E^s)$ and a set of path VNRs $\{P^r_1,P^r_2,...,P^r_n\}$.}
		\Output{Final revenue.}
		\textbf{set} $Flag \leftarrow \textbf{ture}$;\\
		\While{Flag}
		{ \textbf{run} Path Decomposition Phase on $G^s$;\\ 
		  \textbf{run} Embedding by MKP;\\ 
	      \textbf{run} Resource Assignment by MDKP;\\ 
          \If{no path VNR can be embedded}
          {$Flag \leftarrow  \textbf{false}$;\\}
		  \textbf{update} the substrate network $G^s$;\\
		}		
	\end{algorithm}

The time complexity of extracting one substrate path by constructing a $DST$ is $\mathcal{O}(|V^s|+|E^s|)$ and there are at most $|E^s|$ paths extracted. Hence, the total time complexity is $\mathcal{O}(|E^s|\times (|V^s|+|E^s|))$.
The time complexity of embedding by MKP and resource assignment by MDKP, depending on the algorithms for solving MKP and MDKP, are bounded by $T_{MKP}(|E^s|,n)$ and $T_{MDKP}(|V^s|+|E^s|,n)$ respectively \cite{s4:b4}, where $n$ is the number of path VNRs. Network operators can select algorithms according to their computing capability. (\cite{s4:b4} elaborates most of the current algorithms for MKP and MDKP.)
The time complexity for updating the substrate network is $\mathcal{O}(|V^s|+|E^s|)$, and we repeat the \textbf{while}-loop at most $n$ times. Thus, the time complexity of Procedure PE is $\mathcal{O}\big(n \times (|E^s|^2+T_{MKP}(|E^s|,n)+T_{MDKP}(|V^s|+|E^s|,n)\big)$. 


\section{Cycle Embedding}
\label{sec:ce}

For cycle embedding, a substrate network of cycle-like topology is naturally more suitable just like path embedding. Meanwhile, there are many substrate optical rings \cite{s1:b16}. In this section, we investigate the cycle-to-cycle embedding problem, \textit{i.e.}, $Emb(C^s,\{C^r_1,C^r_2,...,C^r_n\})$ where $C^s$ is the substrate cycle and $\{C^r_1,C^r_2,...,C^r_n\}$ is the set of cycle VNRs ($CPU$s and $BW$s are arbitrary). In this paper, we focus on a natural embedding way called \textit{Simplex Cycle Embedding}, which evenly embeds VNs and VLs on the substrate cycle. 



\begin{definition}
\label{def2} \textit{Simplex Cycle Embedding}: We assume the substrate cycle is $C^s(v^s_1,...,v^s_mv^s_1)$. Given a starting SN $v^s_i \in C^s$ with one direction $dir$ (clockwise or anticlockwise denoted by \textbf{"+"} or \textbf{"-"} respectively), we arrange all SNs of $C^s$ in such a sequence, denoted by $Seq(v^s_i,dir)$, that for $dir=\textbf{"+"}$, $[v^s_i, v^s_{i+1}, v^s_{i+2},...,v^s_{i-1}]$, and for $dir=\textbf{"-"}$, $[v^s_i, v^s_{i-1}, v^s_{i-2},...,v^s_{i+1}]$, where all arithmetical operations of subscripts of SNs are modulo $m$. 

Given a cycle VNR $C^r(v^r_1,...,v^r_nv^r_1)$, \textit{Simplex Cycle Embedding} is that: $v^r_j \rightarrow v^s_{i_j}, 1\leq j \leq n$, which can be expressed as $\left( \begin{array}{l}v^r_1~~...~~v^r_{j}~~~v^r_{j+1}~~...~~v^r_n \\v^s_{i_1}~~...~~v^s_{i_j}~~v^s_{i_{j+1}}~~...~~v^s_{i_n}\end{array}\right)$, where the node mapping follows one direction $dir$. More specifically,  $\forall j, v^s_{i_j}$ is ahead of $v^s_{i_{j+1}}$ in $Seq(v^s_{i_1},dir)$. Figure \ref{fig: tsemb} illustrates a triangle VNR embedded on a substrate cycle in clockwise and anticlockwise simplex cycle embeddings, where the numbers beside VLs indicate the demanded $BW$s.

\begin{figure}[!htb]
\centering
\includegraphics[width=0.4\columnwidth]{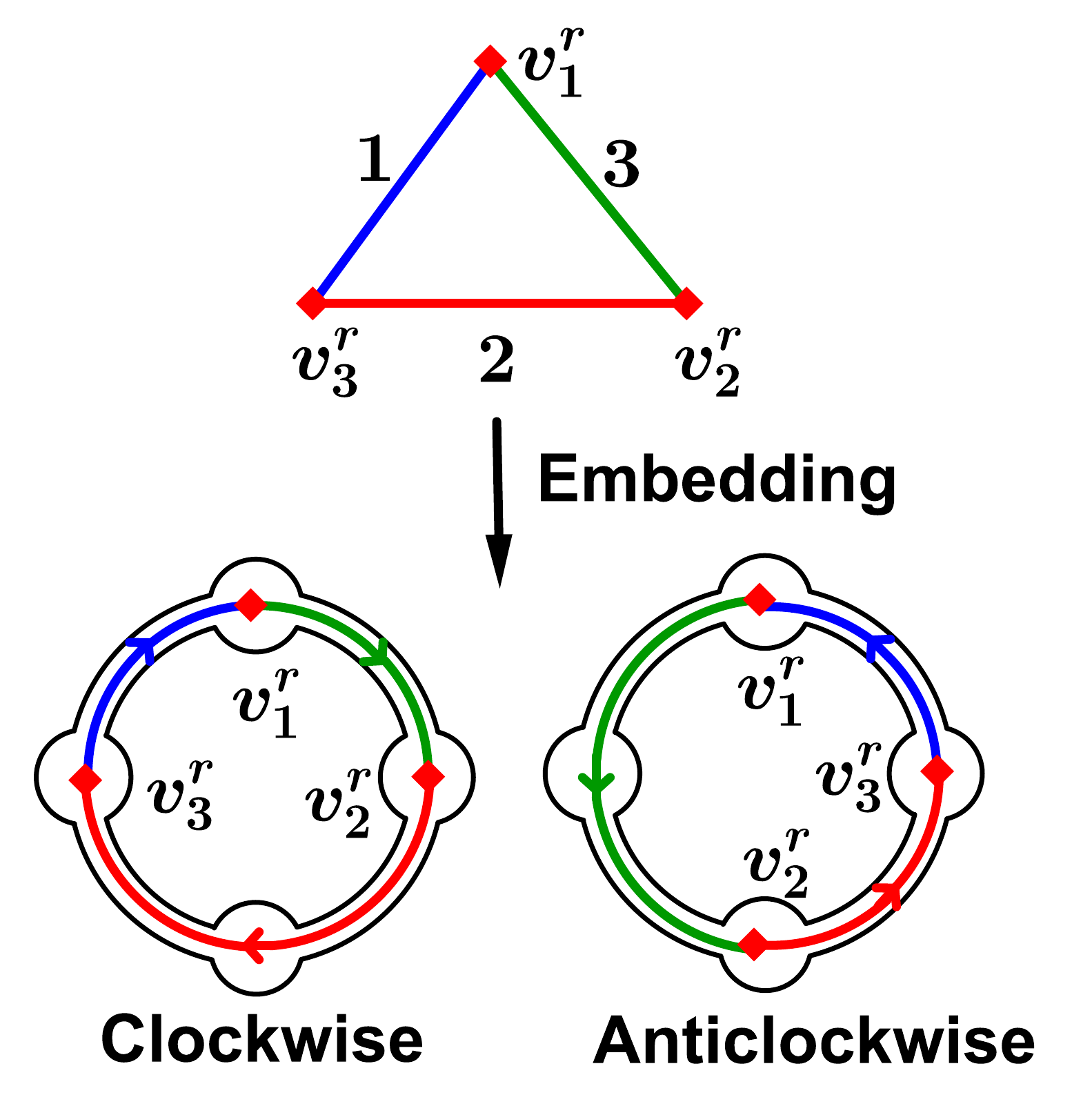}
\caption{Two simplex cycle embeddings: clockwise and anticlockwise.}
\label{fig: tsemb}
\end{figure}
\end{definition}

Similarly, given a $C^s$ and a $C^r$, $Emb(C^s,C^r)$ is the elementary and primary problem in simplex cycle embedding. The similar problem $Emb(G^s,P^r)$ of path embedding is $\mathcal{NP}$-hard even in the uniform setting. Is $Emb(C^s,C^r)$ also $\mathcal{NP}$-hard? Moreover, even the answer of $Emb(C^s,C^r)$ is \textbf{Yes}, different embedding ways could result in different resource consumptions: In Fig. \ref{fig: tsemb}, the $BW$ consumption of the clockwise is $1+2\times 2+3=8$ while that for the anticlockwise is $2\times 3+2+1=9$. Can we efficiently find the least-resource-consuming embedding way? How about the AcR and Rev problems in cycle-to-cycle embedding? 


Following the three important questions, we unfold this section as follows.
\begin{itemize}
\item First, we construct a Weighted Directed Auxiliary Graph (WDAG) in polynomial-time and prove that each of its directed cycles corresponds to a feasible simplex cycle embedding.
\item Then, the minimum weighted directed cycle corresponds to the least-resource-consuming embedding, which can be obtained by dynamic programming in polynomial time.
\item Finally, we prove that both the AcR and Rev problems are strongly $\mathcal{NP}$-hard, and thus devise effective heuristic algorithms to solve them.
\end{itemize}

Given a substrate cycle $C^s(v^s_1,...,v^s_mv^s_1)$ and a cycle VNR $C^r(v^r_1,...,v^r_nv^r_1)$, for all $v^r_j$, let $\mathcal{F}_{v^r_j}=\{v^s_i \in V^s | CPU(v^s_i) \geq CPU(v^r_j) \}$, $\textit{i.e.}$, the set of feasible SNs on which $v^r_j$ can be embedded, and for all $v^r_jv^r_{j+1}$, $\mathcal{F}_{v^r_jv^r_{j+1}}=\{e^s \in E^s| BW(e^s) \geq BW(v^r_jv^r_{j+1}) \}$, \textit{i.e.}, the set of feasible SLs whose $BW$ is not smaller than $v^r_jv^r_{j+1}$'s. If $Emb(C^s,C^r)$ is \textbf{Yes} in simplex cycle embedding, there must exist an embedding way that $v^r_1 \rightarrow v^s_{i_1} \in \mathcal{F}_{v^r_1}$, following one direction $dir$. With respect to the condition that $v^r_1 \rightarrow v^s_{i_1} \in \mathcal{F}_{v^r_1}$ and the embedding direction $dir$, we construct a WDAG denoted by $\hat{G}^{v^s_{i_1}}_{dir}(\hat{V},\hat{A})$, where $\hat{V}$ is the vertex set and $\hat{A}$ is the arc set.


\textbf{1)} The vertex set $\hat{V}$ comprises of $n$ parts $\{\hat{\mathcal{F}}_{v^r_j}|^n_{j=1}\}$, where the $j$-th part $\hat{\mathcal{F}}_{v^r_j}$ corresponds to the set $\mathcal{F}_{v^r_j}$. Except $\hat{\mathcal{F}}_{v^r_1}$, there is a one-to-one mapping, denoted by $MP$, between vertices in $\hat{\mathcal{F}}_{v^r_j}$ and SNs in $\mathcal{F}_{v^r_j}$. In other words, $\forall \hat{v} \in \hat{\mathcal{F}}_{v^r_j}$ one-to-one corresponds to $MP(\hat{v}) \in \mathcal{F}_{v^r_j}$. In $\hat{\mathcal{F}}_{v^r_1}$, there is only one vertex, $\hat{v}_1$, which corresponds to $v^s_{i_1}$, \textit{i.e.}, $MP(\hat{v}_1)=v^s_{i_1}$ as shown in Fig. \ref{fig: indep}.

\textbf{2)} The arc set $\hat{A}$ is iteratively constructed as below: First starting at $\hat{v}_1$, for each such vertex in $\hat{\mathcal{F}}_{v^r_2}$, say $\hat{v}_2$, that satisfies two criteria with $\hat{v}_1$, an arc is constructed with $\hat{v}_1$ as tail and $\hat{v}_2$ as head. The two criteria are as follows. \textbf{Criterion 1:} $MP(\hat{v}_1)$ is ahead of $MP(\hat{v}_2)$ in $Seq(v^s_{i_1},dir)$. \textbf{Criterion 2:} each SL $e^s$, lying in the segment from $MP(\hat{v}_1)$ to $MP(\hat{v}_2)$ following $dir$, belongs to $\mathcal{F}_{v^r_1v^r_2}$. Besides, a weight is assigned to this arc which equals $|MP(\hat{v}_1)-MP(\hat{v}_2)|\times BW(v^r_1v^r_2)$, where $|MP(\hat{v}_1)-MP(\hat{v}_2)|$ is the number of SLs in the segment from $MP(\hat{v}_1)$ to $MP(\hat{v}_2)$ following $dir$. Next, for each such vertex in $\hat{\mathcal{F}}_{v^r_2}$ with incoming edges, say $\hat{v}_2$ and $d^-(\hat{v}_2) > 0$,  we repeat the same procedure on it as we did for $\hat{v}_1$: Searching those vertices in $\hat{\mathcal{F}}_{v^r_3}$ which satisfy the two criteria with $\hat{v}_2$; and arcs are constructed with $\hat{v}_2$ as tail; and weights are computed and assigned to these arcs. After some iterations, at certain vertex part, say $\hat{\mathcal{F}}_{v^r_j}, j<n$, if there is no vertex in $\hat{\mathcal{F}}_{v^r_j}$ whose indegree is greater than 0, then the whole process is terminated. Otherwise, we reach the vertex part $\hat{\mathcal{F}}_{v^r_n}$.  For each such vertex in $\hat{\mathcal{F}}_{v^r_n}$ with non-zero indegree, say $\hat{v}_n$ and $d^-(\hat{v}_n) > 0$, an arc with the corresponding weight is constructed with $\hat{v}_n$ as tail and $\hat{v}_1$ as head, if $\hat{v}_n$ satisfies \textbf{Criterion 2} with $\hat{v}_1$, \textit{i.e.}, each SL $e^s$ lying in the segment from $MP(\hat{v}_n)$ to $MP(\hat{v}_1)$ following $dir$, belongs to $\mathcal{F}_{v^r_nv^r_1}$. Figure \ref{fig: indep} shows a complete WDAG $\hat{G}^{v^s_{i_1}}_{dir}$.

\begin{figure}[!htb]
 \centering
 \includegraphics[width=0.6\columnwidth]{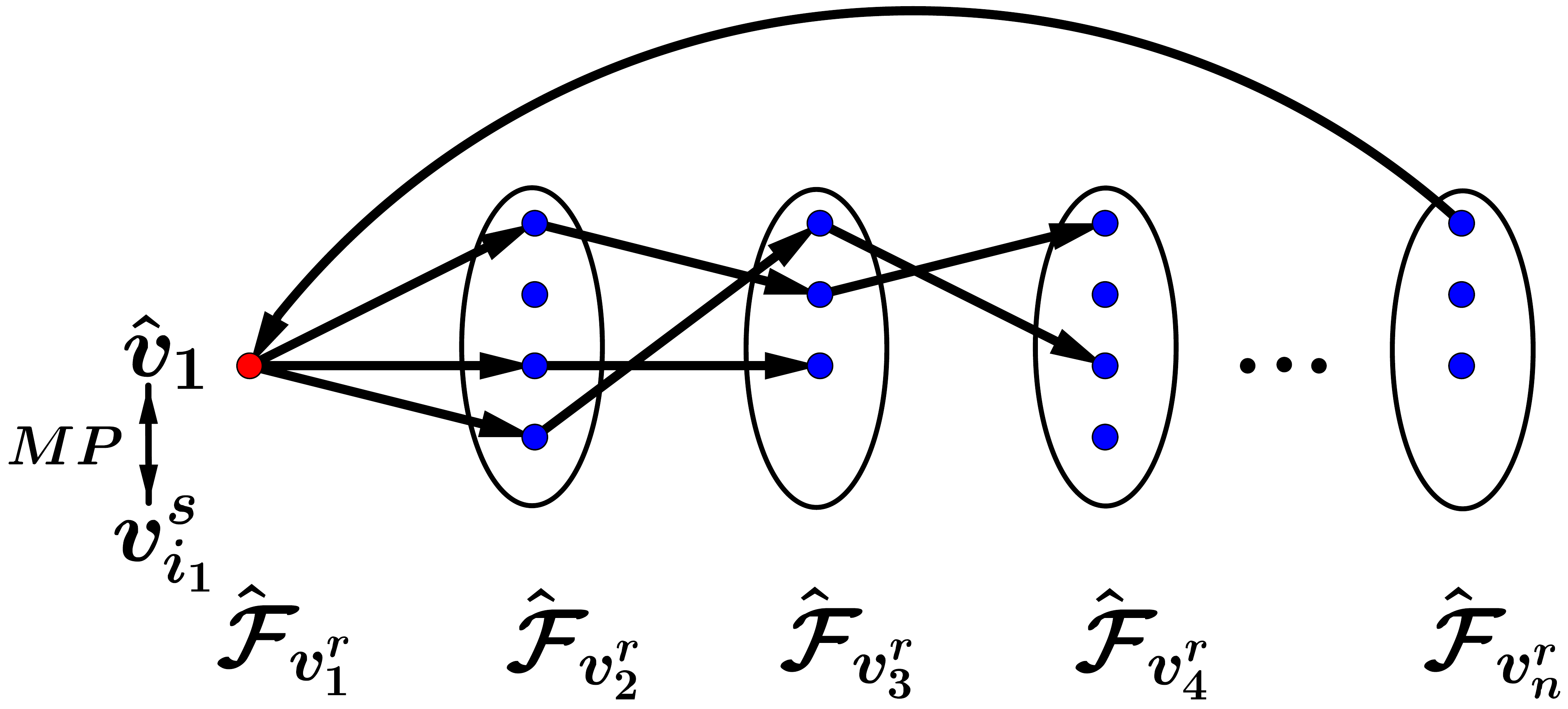}
 \caption{The WDAG $\hat{G}^{v^s_{i_1}}_{dir}$ with respect to $v^s_{i_1}$ and $dir$.}
 \label{fig: indep}
 \end{figure}

\begin{theorem}
\label{the: timecofa2}
The time complexity of constructing a WDAG is $\mathcal{O}(m^3n)$, where $m$ and $n$ are the SN and the VN numbers of the substrate cycle $C^s$ and the cycle VNR $C^r$ respectively.
\end{theorem}

\begin{proof}
The time complexity of constructing a WDAG consists of two parts: 1) establishing the vertex set $\hat{V}$; 2) constructing the weight arc set $\hat{A}$. For the first part, we need to set up the $n$ vertex parts $\hat{\mathcal{F}}_{v^r_j}, \forall 1 \leq j \leq n$, and build a one-to-one mapping $MP$ between vertices in $\hat{\mathcal{F}}_{v^r_j}$ and $\mathcal{F}_{v^r_j}$. Since each $|\hat{\mathcal{F}}_{v^r_j}| \leq m$, time consumption of this part is up to $\mathcal{O}(mn)$. For the second part, in the worst case, for each vertex pair $(\hat{v}_j, \hat{v}_{j+1})$ between $\hat{\mathcal{F}}_{v^r_j}$ and $\hat{\mathcal{F}}_{v^r_{j+1}}$, we need to check \textbf{Criterion 1} and \textbf{Criterion 2} to decide whether to construct an arc. The time consumption of checking \textbf{Criterion 1} and \textbf{Criterion 2} is $\mathcal{O}(m)$. There are at most $m^2n$ arcs, and thus the time consumption of constructing the arc set $\hat{A}$ is $\mathcal{O}(m^3n)$. Combing $\mathcal{O}(mn)$ and $\mathcal{O}(m^3n)$, the total time complexity is $\mathcal{O}(m^3n)$. 
\end{proof}

The WDAG $\hat{G}^{v^s_{i_1}}_{dir}$ has an important property described in Theorem \ref{the: deoto}.

\begin{theorem}
\label{the: deoto}
Given a substrate cycle $C^s(v^s_1,...,v^s_mv^s_1)$ and a cycle VNR $C^r(v^r_1,...,v^r_nv^r_1)$, under the condition that $v^r_1 \rightarrow v^s_{i_1} \in \mathcal{F}_{v^r_1}$ and following $dir$, there is a one-to-one relation between each directed cycle in the WDAG $\hat{G}^{v^s_{i_1}}_{dir}$ and each feasible simplex cycle embedding way.
\end{theorem}

\begin{proof}
For each feasible simplex cycle embedding way under the condition that $v^r_1 \rightarrow v^s_{i_1}$ and following $dir$, \textit{w.l.o.g.}, we assume it is $\left( \begin{array}{l}v^r_1~~~v^r_{2}~~...~~v^r_{j}~~~v^r_{j+1}~~...~~v^r_n \\v^s_{i_1}~~v^s_{i_2}~~...~~v^s_{i_j}~~v^s_{i_j+1}~~...~~v^s_{i_n}\end{array}\right)$, where each $v^r_j \rightarrow v^s_{i_j}$ and $v^s_{i_j} \in \mathcal{F}_{v^r_j}$. Since $MP$ is a one-to-one mapping from $\hat{\mathcal{F}_{v^r_j}}$ to $\mathcal{F}_{v^r_j}$, we use $MP^{-1}$ to represent the inverse, \textit{i.e.}, $MP^{-1}(v^s_{i_j}) \in \hat{\mathcal{F}_{v^r_j}}, MP(MP^{-1}(v^s_{i_j}))=v^s_{i_j}$. As it is feasible,  $MP^{-1}(v^s_{i_{j+1}})$ must satisfy \textbf{Criterion 1} and \textbf{Criterion 2} with $MP^{-1}(v^s_{i_j})$. Therefore, according to the construction process, there is an arc with $MP^{-1}(v^s_{i_j})$ as tail and $MP^{-1}(v^s_{i_{j+1}})$ as head, and {\footnotesize$\left(MP^{-1}(v^s_{i_1}),...,MP^{-1}(v^s_{i_j}),...,MP^{-1}(v^s_{i_n})MP^{-1}(v^s_{i_1})\right)$} forms a directed cycle in $\hat{G}^{v^s_{i_1}}_{dir}$.

For each directed cycle in $\hat{G}^{v^s_{i_1}}_{dir}$, say $\hat{C}$, according to the construction process, $\hat{C}$ must pass through exactly one vertex in each vertex part $\hat{F}_{v^r_j}, \forall 1 \leq j \leq n$, say $\left(\hat{v}_1,...,\hat{v}_j,...,\hat{v}_n\right)$. Since $\hat{v}_{j+1}$ satisfies \textbf{Criterion 1} and \textbf{Criterion 2} with $\hat{v}_{j}$, it is obvious that 

$MP(\hat{C})=\left( \begin{array}{l}v^r_1~~~........~~~~v^r_{j}~~~~~.......~~v^r_n \\ MP(\hat{v}_1)...MP(\hat{v}_j)...MP(\hat{v}_n)\end{array}\right)$ must be a feasible simplex cycle embedding way. 
\end{proof}

Subsequently, combining with the weights assigned to arcs of the WDAG $\hat{G}^{v^s_{i_1}}_{dir}$, the sum of arc weights of a directed cycle in $\hat{G}^{v^s_{i_1}}_{dir}$ is equal to the total $BW$ consumption of the embedding way which corresponds to the directed cycle. Since the $CPU$ consumption is fixed to the sum of all demanded $CPUs$ of VNs, to obtain the least-resource-consuming embedding way, we just need to search the minimum weighted directed cycle in the WDAG $\hat{G}^{v^s_{i_1}}_{dir}$ which can be solved by dynamic programming in polynomial time.

To obtain the optimal least-resource-consuming simplex cycle embedding, we just need to construct $2\times |\mathcal{F}_{v^r_1}|$ WDAGs (two directions), search the directed cycle with the minimum weigth in each WDAG, and output the smallest one among them. We formally give Algorithm \ref{al: C2CE}.
\begin{algorithm}[!h]
    \caption{Procedure of Cycle-to-Cycle Embedding (C2CE)}
    \label{al: C2CE}
		\SetKwInOut{Input}{Input}
		\SetKwInOut{Output}{Output}
		\SetKw{KwAnd}{and}
		\SetKw{KwSuch}{s.t.}
		\Input{$C^s(v^s_1,...,v^s_mv^s_1)$ and $C^r(v^r_1,...,v^r_nv^r_1)$}
		\Output{The least-resource-consuming embedding way.}
		\textbf{set}  $\mathcal{F}_{v^r_j} \leftarrow\{v^s_i| CPU(v^s_i) \geq CPU(v^r_j) \}, \forall j$;\\
        \textbf{set}  $\mathcal{F}_{v^r_jv^r_{j+1}} \leftarrow\{v^s_iv^s_{i+1}| BW(v^s_iv^s_{i+1}) \geq BW(v^r_jv^r_{j+1})\}$;\\
         \textbf{set} $EMB \leftarrow \emptyset$, $Cost \leftarrow \infty$;\\
		\For{$v^r_{i_1} \in \mathcal{F}_{v^r_1}$}
		{ \For{$dir$ (\textbf{"+"} or \textbf{"-"})}
           { \textbf{construct} the WDAG;\\
             \textbf{search} the minimum weighted directed cycle $\hat{C}$ in the                       WDAG;\\
             \If{$w(\hat{C}) < Cost$}
              {$EMB \leftarrow MP(\hat{C})$, $Cost  \leftarrow w(\hat{C})$;\\}
           }
     }	
        \textbf{return} $EMB$;
	\end{algorithm}
 
In \textit{Lines 1-2}, for each VN $v^r_j$ and VL $v^r_jv^r_{j+1}$, we set up the feasible SN sets and SL sets. The time complexity of \textit{Lines 1-2} is $\mathcal{O}(mn)$. At \textit{Lines 3}, we set a variable $EMB$ to record the optimal embedding way whose initial value is $\emptyset$ and another variable $Cost$ to record the $BW$ consumption of the $EMB$ whose initial value is large enough denoted by $\infty$. In \textit{Lines 4-9}, for each $v^s_{i_1} \in \mathcal{F}_{v^r_1}$ and each direction $dir$ (\textbf{"+"} or \textbf{"-"}), we construct the corresponding WDAG at \textit{Line 6} whose time complexity is $\mathcal{O}(m^3n)$ by Theorem \ref{the: timecofa2}. At \textit{Line 7}, we search the minimum weighted directed cycle $\hat{C}$, and its time complexity is $\mathcal{O}(m^2n)$ by dynamic programing. At \textit{Lines 8-9}, if the weight of $\hat{C}$ denoted by $w(\hat{C})$ is less than current $Cost$, we replace $Cost$ by $w(\hat{C})$ and $EMB$ by the embedding way denoted by $MP(\hat{C})$ which corresponds to $\hat{C}$. Finally, at \textit{Line 10}, we return the $EMB$ (if $EMB=\emptyset$ then $C^r$ can not be embedded on $C^s$ by simplex cycle embedding). The total time complexity of Algorithm \ref{al: C2CE} is $\mathcal{O}(mn)+2m\times \left(\mathcal{O}(m^3n)+\mathcal{O}(m^2n)\right)+\mathcal{O}(1)=\mathcal{O}(m^4n)$, where $m$ and $n$ are the SN and the VN numbers of the substrate cycle $C^s$ and the cycle VNR $C^r$ respectively.


For the $Emb(C^s,C^r)$ problem in simplex cycle embedding, we can solve it in polynomial time. How about the AcR and Rev problem? Unfortunately, both of them are still $\mathcal{NP}$-hard.

\begin{theorem}
\label{the: snpc2c}
 In cycle-to-cycle embedding, both the AcR and Rev problems are $\mathcal{NP}$-hard. Moreover, the hardnesses of each problem is no less than any $d$-DKP, where $d$ is any constant integer.
\end{theorem}

\begin{proof}
Since the AcR problem is a special case of the Rev problem, we just need to prove the AcR problem is $\mathcal{NP}$-hard. To this end, we polynomial-timely reduce the $\mathcal{NP}$-hard problem "Cardinality $d$-DKP" \cite{s4:b4} to the AcR problem. Cardinality $d$-DKP is a special $d$ dimensional MDKP, \textit{i.e.}, the knapsack is with a $d$-dimensional capacity attributes ($b_1$,...,$b_i$,...,$b_d$)
and each $j$-th item is with a $d$-dimensional size attributes $(s_{j1},...,s_{ji},...,s_{jd})$. The objective is to maximize the number of packed items.

Given an instance in Cardinality $d$-DKP, we construct the substrate cycle in such way: There are $d$ SNs in the substrate cycle $C^s(v^s_1,...,v^s_i,...,v^s_d)$, and $CPU(v^s_1)=b_1$ and $CPU(v^s_i)=B_i\times b_i, \forall 1<i\leq d$ where $B_i$ are relatively big numbers explained later. Each SL is with $BW=n$, \textit{i.e.}, the number of items. We construct $n$ cycle VNRs in such way: For the $j$-th cycle VNR $C^r(v^r_{j1},...,v^r_{ji},...,v^r_{jd})$, there are $d$ VNs, and $CPU(v^r_{j1})=s_{j1}$ and $CPU(v^r_{ji})=B_i\times s_{ji}, \forall 1<i\leq d$ such that $B_i\times s_{ji} > \max_{1<k<i}(B_k\times b_k, b_1)$ (by setting $CPUs$ like this the $v^r_{ji}$ can only be embedded on $v^s_i$). Each VL is with $BW=1$. Thus, the solution of the instance of Cardinality $d$-DKP is equivalent to that of the AcR problem.
\end{proof}

As shown in \cite{s4:b4}, even the $2$-DKP is strongly $\mathcal{NP}$-hard and the hardness of solving $d$-DKP keeps entrenched with the increase of $d$. To effectively solve the strongly $\mathcal{NP}$-hard problem, we herein develop a heuristic algorithm based on the optimization for single cycle embedding as follows. Intuitively, for a cycle $C^r$, if its ratio of revenue to resource consumption is higher than the others, it tends to be embedded so as to achieve a more efficient income for the InP. This consists of the main motivation of our greedy strategy in Algorithm \ref{al: revenuc2c}: Given a substrate cycle $C^s$ and a set of cycle VNRs $\{C^r_1,C^r_2,...,C^r_n\}$, for each $C^r_j$, we first estimate the ratio of revenue to resource consumption, \textit{i.e.}, $\frac{w_j}{\sum_{v^r \in C^r_j}CPU(v^r)+\sum_{e^r \in C^r_j}BW(e^r)}$. We then arrange them in the descending order of the estimated ratios, and sequentially embed them on the $C^s$ by procedure C2CE until no more cycle VNR can be embedded by simplex cycle embedding. However, one thing should be noted that the simplex cycle embedding has its own shortage, \textit{i.e.}, it limits the solution space. Therefore, if no cycle VNR can be embedded by simplex cycle embedding, we continue the embedding by running general algorithms. Via this combination, both merits of simplex cycle embedding and general algorithms can be conflated. 
\begin{algorithm}[!h]
    \caption{Procedure of Greedy Revenue (GR)}
    \label{al: revenuc2c}
		\SetKwInOut{Input}{Input}
		\SetKwInOut{Output}{Output}
		\SetKw{KwAnd}{and}
		\SetKw{KwSuch}{s.t.}
		\Input{$C^s$ and $\{C^r_1,C^r_2,...,C^r_n\}$.}
		\Output{The final revenue.}
        \textbf{set}  cycle VNRs in descending order by estimated ratios;\\
        \textbf{run} C2CE to sequentially embed cycle VNRs until can't;\\     
        \textbf{run} general algorithms;\\    
        \textbf{return} final revenue;
	\end{algorithm}
 

\section{Numerical Results}
\label{sec:nr}
In this section, we compare our proposed algorithms PE for path embedding and GR for cycle embedding to the existing general algorithms. Two general embedding algorithms from \cite{s1:b7} and \cite{s1:b4}, denoted by RW and BA respectively, are used as our benchmarks. Here, we use GRRW and GRBA to denote the procedures GR which invoke RW and BA respectively. We repeat each simulation 50 times under the same circumstance to ensure sufficient statistical accuracy, and a 95\% confidence interval is given to each numerical result. All the simulations have been run by MATLAB 2015a on a computer with 3.2 GHz Intel(R) Core(TM) i5-4690S CPU and 8 GBytes RAM. 

\subsection{Evaluation Environments}

\subsubsection{Path Embedding}
\hfill
\paragraph{Substrate Networks}
We use the GT-ITM tool \cite{s5:b1}, prevailing in the generation of practical network topologies, to randomly generate two substrate networks denoted by $G^s_1$, $G^s_2$ respectively. Both of these substrate networks have 100 SNs and 1000 SLs, corresponding to a medium-sized ISP. Besides, we also use a complete graph of 100 SNs as the substrate network denoted by $CG^s$. The $CPU$ and $BW$ of each SN and SL are set as 100 units.
\paragraph{Virtual Network Requests}
The length of each path VNR is randomly generated in the range of $[5,10]$. The $CPU$ and $BW$ of each VN and VL are randomly generated in the range of $[1,5]$ units. The number of path VNRs is set as $1000$ in each simulation. 

\subsubsection{Cycle Embedding}
\hfill
\paragraph{Substrate Networks} We set up three substrate cycles denoted by $C^s_{20}$, $C^s_{25}$, $C^s_{30}$ respectively, whose number of SNs are 20, 25 and 30 respectively,  corresponding to the sizes of existing substrate optical rings. The $CPU$ and $BW$ of each SN and SL are set as 100 units.

\paragraph{Virtual Network Requests} The number of VNs of each cycle VNR is randomly generated in the range of $[5,10]$. The $CPU$ and $BW$ of each VN and VL are randomly generated in the range of $[1,5]$ units. The number of cycle VNRs is set as $100$ in each simulation. 

\subsubsection{Performance Metrics}
\hfill

The comparisons are performed for both the AcR and Rev problems.
\begin{itemize}
\item The AcR problem: The revenue of each VNR is set to be one. Besides, we tweak the objective function of the AcR problem as $\frac{|S|}{n}$, where $S$ is the subset of embedded VNRs and $n$ is the number of total VNRs.

\item The Rev problem: The revenue of each VNR is proportional to its VN number in the range of [5,10].
\end{itemize}

\subsection{Simulation Results}

\subsubsection{Path Embedding}
\hfill

Figures \ref{fig:nacrp} and \ref{fig:nrRev} respectively demonstrate the numerical results of the AcR and Rev problems in path embedding. The average acceptance ratio of all the three substrate networks, as shown in Fig \ref{subfig:acrpe}, is $41.06\%$ for PE, compared to $30.70\%$ and $29.75\%$ for RW and BA respectively. The average revenue, as shown in Fig. \ref{subfig:revpe}, is $3052.12$ for PE, compared to $2308.05$ and $2251.01$ for RW and BA respectively. For time complexity as shown in Figs. \ref{subfig:acrtpe} and \ref{subfig:revtpe}, PE is of  an average run time of $1.32s$, obviously smaller than that of RW and BA, $20.44s$ and $8.28s$ repsectively.
In summary, PE is much more efficient and effective than the two general algorithms RW and BA in path embedding.

\begin{figure}[!htbp]
\centering
\begin{subfigure}[Acceptance ratio]{
\label{subfig:acrpe}
\begin{tikzpicture}[scale=0.7]
      \begin{axis}[
      width  = 0.5*\textwidth,
      height = 8cm,
      major x tick style = transparent,
      ybar=5*\pgflinewidth,
      bar width=12pt,
      symbolic x coords={\LARGE{$G^s_1$},\LARGE{$G^s_2$},\LARGE{$CG^s$}},
      xtick = data,
      scaled y ticks = false,
      enlarge x limits=0.25,
      ymin=0.2,
      ylabel= \LARGE{Acceptance Ratio},
      legend entries={\Large{PE},\Large{RW},\Large{BA}},
      legend style={at={(0.035,0.9)},anchor= west,legend cell align=left,legend columns=2},
      yticklabel=\pgfmathparse{100*\tick}\pgfmathprintnumber{\pgfmathresult}\,\%,
      yticklabel style={/pgf/number format/.cd,fixed,precision=2}
    ]

             \addplot [fill=red,area legend,
error bars/.cd, y dir=both, y explicit]
           coordinates {
          (\LARGE{$G^s_1$},0.3979) += (0,0.0158) -= (0,0.0158)
          (\LARGE{$G^s_2$},0.3978) += (0,0.016) -= (0,0.016)
          (\LARGE{$CG^s$},0.4362) += (0,0.017) -= (0,0.017)
         
          };

   \addplot [fill=blue,area legend,
error bars/.cd, y dir=both, y explicit]
           coordinates {
          (\LARGE{$G^s_1$},0.2995) += (0,0.01) -= (0,0.01)
          (\LARGE{$G^s_2$},0.31) += (0,0.008) -= (0,0.008)
          (\LARGE{$CG^s$},0.3116) += (0,0.0109) -= (0,0.0109)
         
          };

    \addplot [fill=green,area legend,
error bars/.cd, y dir=both, y explicit]
           coordinates {
          (\LARGE{$G^s_1$},0.2893) += (0,0.009) -= (0,0.009)
          (\LARGE{$G^s_2$},0.3014) += (0,0.008) -= (0,0.008)
          (\LARGE{$CG^s$},0.3018) += (0,0.011) -= (0,0.011)
         
          }; 
           
  \end{axis}
  \end{tikzpicture}}
\end{subfigure}
\begin{subfigure}[Time complexity]{
\label{subfig:acrtpe}
\begin{tikzpicture}[scale=0.7]
      \begin{axis}[
      width  = 0.5*\textwidth,
      height = 8cm,
      major x tick style = transparent,
      ybar=5*\pgflinewidth,
      bar width=12pt,
      symbolic x coords={\LARGE{$G^s_1$},\LARGE{$G^s_2$},\LARGE{$CG^s$}},
      xtick = data,
      scaled y ticks = false,
      enlarge x limits=0.25,
      ymin=0,
      ylabel= \LARGE{Time($s$)},
      legend entries={\Large{PE},\Large{RW},\Large{BA}},
      legend style={at={(0.035,0.9)},anchor= west,legend cell align=left,legend columns=2},
      yticklabel style={/pgf/number format/.cd,fixed,precision=2}
  ]

             \addplot [fill=red,area legend,
error bars/.cd, y dir=both, y explicit]
           coordinates {
          (\LARGE{$G^s_1$},1.1891) += (0,0.3173) -= (0,0.3173)
          (\LARGE{$G^s_2$},1.0413) += (0,0.3091) -= (0,0.3091)
          (\LARGE{$CG^s$},1.7152) += (0,0.4) -= (0,0.4)
         
          };

   \addplot [fill=blue,area legend,
error bars/.cd, y dir=both, y explicit]
           coordinates {
          (\LARGE{$G^s_1$},15.4141) += (0,0.3937) -= (0,0.3937)
          (\LARGE{$G^s_2$}, 16.3184) += (0,0.4) -= (0,0.4)
          (\LARGE{$CG^s$},29.5953) += (0,0.5) -= (0,0.5)
         
          };

    \addplot [fill=green,area legend,
error bars/.cd, y dir=both, y explicit]
           coordinates {
          (\LARGE{$G^s_1$},7.3457) += (0,0.5034) -= (0,0.5034)
          (\LARGE{$G^s_2$},7.9145) += (0,0.4) -= (0,0.4)
          (\LARGE{$CG^s$},9.5990) += (0,0.6) -= (0,0.6)
         
          }; 
           
  \end{axis}
  \end{tikzpicture}}
  \end{subfigure}
  \caption{Numerical results of the AcR problem in path embedding.}

\label{fig:nacrp}
\end{figure}
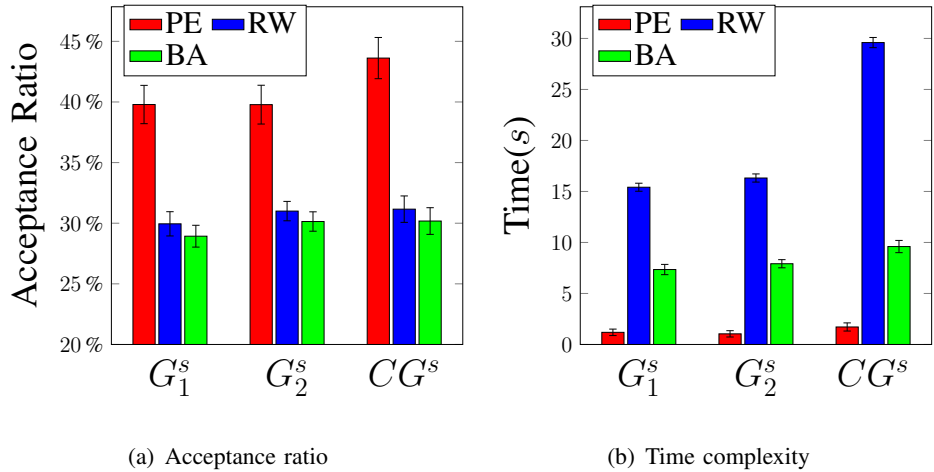

\begin{figure}[!htbp]
\centering
\begin{subfigure}[Revenue]{
\label{subfig:revpe}
\begin{tikzpicture}[scale=0.7]
      \begin{axis}[
      width  = 0.50*\textwidth,
      height = 8cm,
      major x tick style = transparent,
      ybar=5*\pgflinewidth,
      bar width=12pt,
      symbolic x coords={\LARGE{$G^s_1$},\LARGE{$G^s_2$},\LARGE{$CG^s$}},
      xtick = data,
      scaled y ticks = false,
      enlarge x limits=0.25,
      ymin=2000,
      ylabel= \LARGE{Revenue},
      legend entries={\Large{PE},\Large{RW},\Large{BA}},
      legend style={at={(0.035,0.9)},anchor= west,legend cell align=left,legend columns=2},
      yticklabel style={/pgf/number format/.cd,fixed,precision=2}
  ]

             \addplot [fill=red,area legend,
error bars/.cd, y dir=both, y explicit]
           coordinates {
          (\LARGE{$G^s_1$},2967.36) += (0,100.25) -= (0,100.25)
          (\LARGE{$G^s_2$},2980.00) += (0,99) -= (0,0.99)
          (\LARGE{$CG^s$},3209) += (0,104.67) -= (0,104.67)
         
          };

   \addplot [fill=blue,area legend,
error bars/.cd, y dir=both, y explicit]
           coordinates {
          (\LARGE{$G^s_1$},2250.76) += (0,95.50) -= (0,95.50)
          (\LARGE{$G^s_2$}, 2281.40) += (0,97.23) -= (0,97.23)
          (\LARGE{$CG^s$},2391.99) += (0,100.1) -= (0,100.1)
         
          };

    \addplot [fill=green,area legend,
error bars/.cd, y dir=both, y explicit]
           coordinates {
          (\LARGE{$G^s_1$},2194.64) += (0,89.50) -= (0,89.50)
          (\LARGE{$G^s_2$},2258.60) += (0,92.87) -= (0,92.87)
          (\LARGE{$CG^s$},2299.78) += (0,101.32) -= (0,101.32)
         
          }; 
           
  \end{axis}
  \end{tikzpicture}}
  \end{subfigure}
\begin{subfigure}[Time complexity]{
\label{subfig:revtpe}
\begin{tikzpicture}[scale=0.7]
      \begin{axis}[
      width  = 0.50*\textwidth,
      height = 8cm,
      major x tick style = transparent,
      ybar=5*\pgflinewidth,
      bar width=12pt,
      symbolic x coords={\LARGE{$G^s_1$},\LARGE{$G^s_2$},\LARGE{$CG^s$}},
      xtick = data,
      scaled y ticks = false,
      enlarge x limits=0.25,
      ymin=0,
      ylabel= \LARGE{Time($s$)},
      legend entries={\Large{PE},\Large{RW},\Large{BA}},
      legend style={at={(0.035,0.9)},anchor= west,legend cell align=left,legend columns=2},
      yticklabel style={/pgf/number format/.cd,fixed,precision=2}
  ]

             \addplot [fill=red,area legend,
error bars/.cd, y dir=both, y explicit]
           coordinates {
          (\LARGE{$G^s_1$},1.40) += (0,0.52) -= (0,0.52)
          (\LARGE{$G^s_2$},1.31) += (0,0.45) -= (0,0.45)
          (\LARGE{$CG^s$},2.03) += (0,0.51) -= (0,0.51)
         
          };

   \addplot [fill=blue,area legend,
error bars/.cd, y dir=both, y explicit]
           coordinates {
          (\LARGE{$G^s_1$},17.64) += (0,0.88) -= (0,0.88)
          (\LARGE{$G^s_2$}, 17.53) += (0,0.81) -= (0,0.81)
          (\LARGE{$CG^s$},31.92) += (0,0.89) -= (0,0.89)
         
          };

    \addplot [fill=green,area legend,
error bars/.cd, y dir=both, y explicit]
           coordinates {
          (\LARGE{$G^s_1$},8.87) += (0,0.56) -= (0,0.56)
          (\LARGE{$G^s_2$},8.89) += (0,0.47) -= (0,0.47)
          (\LARGE{$CG^s$},10.1) += (0,0.61) -= (0,0.61)
         
          }; 
           
  \end{axis}
  \end{tikzpicture}}
  \end{subfigure}
\caption{Numerical results of the Rev problem in path embedding.}
\label{fig:nrRev}
\end{figure}
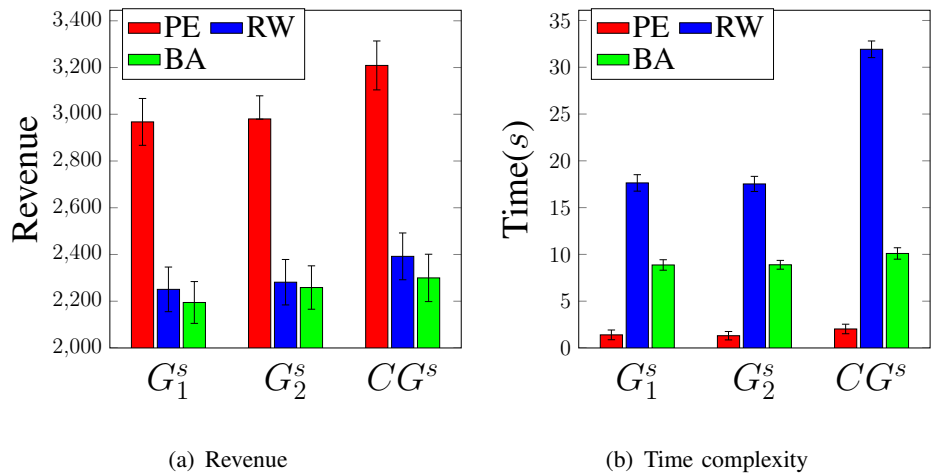

\subsubsection{Cycle Embedding}
\hfill

Figures \ref{fig:nacrpce} and \ref{fig:nrRevce} respectively showcase the numerical results of the AcR and Rev problems in cycle embedding. The average acceptance ratio, as shown in Fig. \ref{subfig:acrce}, of GRRW and GRBA are $31.13\%$ and $30.31\%$ respectively compared to $25.63\%$ and $24.84\%$ of RW and BA respectively. The average revenue, as shown in Fig. \ref{subfig:revce}, of GRRW and GRBA are $239.68$ and $234.96$ respectively while $189.80$ and $185.62$ of RW and BA respectively. From the aspect of final results of acceptance ratios and revenues, GRRW and GRBA take advantage over RW and BA. For the time complexity as shown in Figs. \ref{subfig:acrtce} and \ref{subfig:revtce}, while the run times of RW and BA are relatively stable and smaller than $2.5s$, that of GRRW and GRBA are quickly climbing because the time complexity of construction of the WDAG is fourth-order about the number of SNs. But, the corresponding acceptance ratios and revenues do not improve much with the increase of run times of GRRW and GRBA. Thus, in general substrate networks, in the future more work is needed to balance the size of decomposed substrate cycles and develop decomposition strategies so as to constitute cycle embedding algorithms as competitive as PE for path embedding.

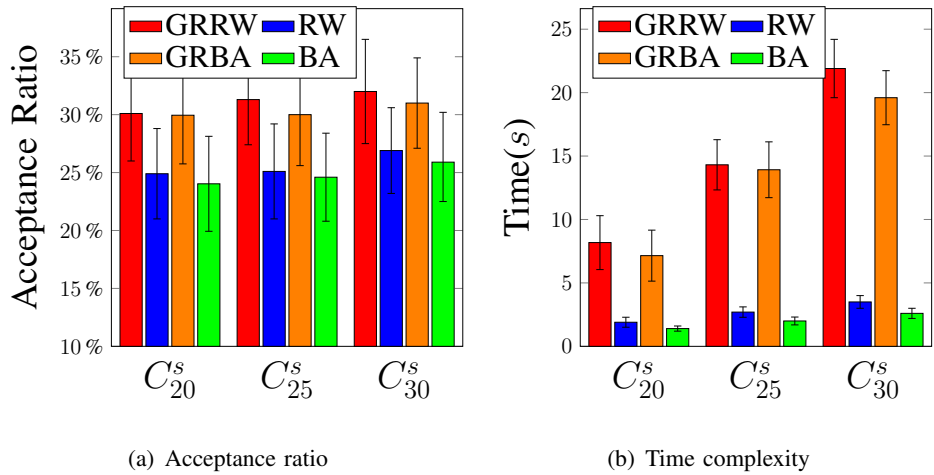
\begin{figure}[!htbp]
\centering
\begin{subfigure}[Acceptance ratio]{
\label{subfig:acrce}
\begin{tikzpicture}[scale=0.7]
      \begin{axis}[
      width  = 0.5*\textwidth,
      height = 8cm,
      major x tick style = transparent,
      ybar=5*\pgflinewidth,
      bar width=12pt,
      symbolic x coords={\LARGE{$C^s_{20}$},\LARGE{$C^s_{25}$},\LARGE{$C^s_{30}$}},
      xtick = data,
      scaled y ticks = false,
      enlarge x limits=0.25,
      ymin=0.1,
      ylabel= \LARGE{Acceptance Ratio},
      legend entries={\Large{GRRW},\Large{RW},\Large{GRBA},\Large{BA}},
      legend style={at={(0.035,0.9)},anchor= west,legend cell align=left,legend columns=2},
      yticklabel=\pgfmathparse{100*\tick}\pgfmathprintnumber{\pgfmathresult}\,\%,
      yticklabel style={/pgf/number format/.cd,fixed,precision=2}
  ]

             \addplot [fill=red,area legend,
error bars/.cd, y dir=both, y explicit]
           coordinates {
          (\LARGE{$C^s_{20}$},0.301) += (0,0.041) -= (0,0.041)
          (\LARGE{$C^s_{25}$},0.313) += (0,0.039) -= (0,0.039)
          (\LARGE{$C^s_{30}$},0.32) += (0,0.045) -= (0,0.045)
         
          };

           \addplot [fill=blue,area legend,
error bars/.cd, y dir=both, y explicit]
           coordinates {
          (\LARGE{$C^s_{20}$},0.249) += (0,0.039) -= (0,0.039)
          (\LARGE{$C^s_{25}$},0.251) += (0,0.041) -= (0,0.041)
          (\LARGE{$C^s_{30}$},0.269) += (0,0.037) -= (0,0.037)
          
          };

   \addplot [fill=orange,area legend,
error bars/.cd, y dir=both, y explicit]
           coordinates {
          (\LARGE{$C^s_{20}$},0.2995) += (0,0.042) -= (0,0.042)
          (\LARGE{$C^s_{25}$},0.30) += (0,0.044) -= (0,0.044)
          (\LARGE{$C^s_{30}$},0.31) += (0,0.039) -= (0,0.039)
         
          };

    \addplot [fill=green,area legend,
error bars/.cd, y dir=both, y explicit]
           coordinates {
          (\LARGE{$C^s_{20}$},0.2403) += (0,0.041) -= (0,0.041)
          (\LARGE{$C^s_{25}$},0.246) += (0,0.038) -= (0,0.038)
          (\LARGE{$C^s_{30}$},0.259) += (0,0.043) -= (0,0.034)
         
          }; 
           
  \end{axis}
  \end{tikzpicture}}
\end{subfigure}
\begin{subfigure}[Time complexity]{
\label{subfig:acrtce}
\begin{tikzpicture}[scale=0.7]
      \begin{axis}[
      width  = 0.5*\textwidth,
      height = 8cm,
      major x tick style = transparent,
      ybar=5*\pgflinewidth,
      bar width=12pt,
      symbolic x coords={\LARGE{$C^s_{20}$},\LARGE{$C^s_{25}$},\LARGE{$C^s_{30}$}},
      xtick = data,
      scaled y ticks = false,
      enlarge x limits=0.25,
      ymin=0,
      ylabel= \LARGE{Time($s$)},
      legend entries={\Large{GRRW},\Large{RW},\Large{GRBA},\Large{BA}},
      legend style={at={(0.035,0.9)},anchor= west,legend cell align=left,legend columns=2},
      yticklabel style={/pgf/number format/.cd,fixed,precision=2}
  ]

             \addplot [fill=red,area legend,
error bars/.cd, y dir=both, y explicit]
           coordinates {
          (\LARGE{$C^s_{20}$},8.18) += (0,2.12) -= (0,2.13)
          (\LARGE{$C^s_{25}$},14.31) += (0,1.98) -= (0,1.98)
          (\LARGE{$C^s_{30}$},21.9) += (0,2.3) -= (0,2.3)
         
          };

           \addplot [fill=blue,area legend,
error bars/.cd, y dir=both, y explicit]
           coordinates {
          (\LARGE{$C^s_{20}$}, 1.9) += (0,0.3945) -= (0,0.3945)
          (\LARGE{$C^s_{25}$},2.7) += (0,0.41) -= (0,0.41)
          (\LARGE{$C^s_{30}$},3.5) += (0,0.5) -= (0,0.5)
          
          };

   \addplot [fill=orange,area legend,
error bars/.cd, y dir=both, y explicit]
           coordinates {
          (\LARGE{$C^s_{20}$},7.15) += (0,2.01) -= (0,2.01)
          (\LARGE{$C^s_{25}$}, 13.92) += (0,2.2) -= (0,2.2)
          (\LARGE{$C^s_{30}$},19.6) += (0,2.13) -= (0,2.13)
         
          };

    \addplot [fill=green,area legend,
error bars/.cd, y dir=both, y explicit]
           coordinates {
          (\LARGE{$C^s_{20}$},1.4) += (0,0.2034) -= (0,0.2034)
          (\LARGE{$C^s_{25}$},2.0) += (0,0.31) -= (0,0.31)
          (\LARGE{$C^s_{30}$},2.6) += (0,0.4) -= (0,0.4)
         
          }; 
           
  \end{axis}
  \end{tikzpicture}}
  \end{subfigure}
  \caption{Numerical results of the AcR problem in cycle embedding.}

\label{fig:nacrpce}
\end{figure}

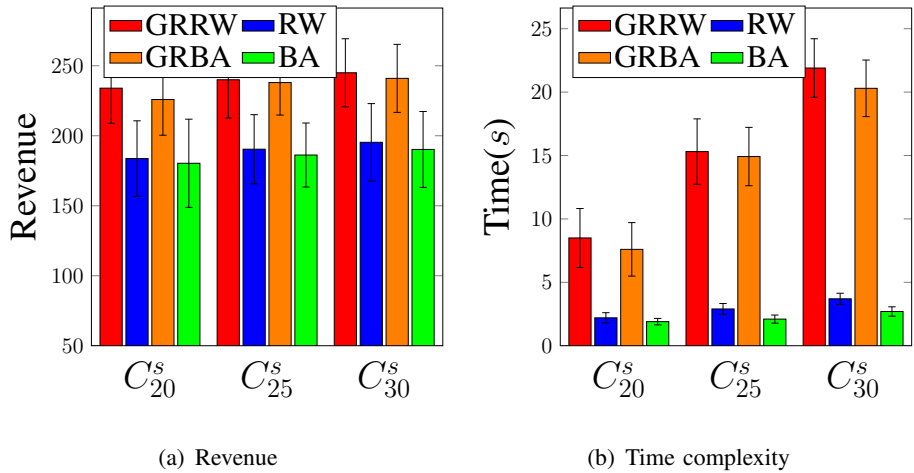
\begin{figure}[!htbp]
\centering
\begin{subfigure}[Revenue]{
\label{subfig:revce}
\begin{tikzpicture}[scale=0.7]
      \begin{axis}[
      width  = 0.50*\textwidth,
      height = 8cm,
      major x tick style = transparent,
      ybar=5*\pgflinewidth,
      bar width=12pt,
      symbolic x coords={\LARGE{$C^s_{20}$},\LARGE{$C^s_{25}$},\LARGE{$C^s_{30}$}},
      xtick = data,
      scaled y ticks = false,
      enlarge x limits=0.25,
      ymin=50,
      ylabel= \LARGE{Revenue},
      legend entries={\Large{GRRW},\Large{RW},\Large{GRBA},\Large{BA}},
      legend style={at={(0.035,0.9)},anchor= west,legend cell align=left,legend columns=2},
      yticklabel style={/pgf/number format/.cd,fixed,precision=2}
  ]

             \addplot [fill=red,area legend,
error bars/.cd, y dir=both, y explicit]
           coordinates {
          (\LARGE{$C^s_{20}$},234) += (0,25.1) -= (0,25.1)
          (\LARGE{$C^s_{25}$},240) += (0,27.32) -= (0,27.32)
          (\LARGE{$C^s_{30}$},245) += (0,24.31) -= (0,24.31)
         
          };

           \addplot [fill=blue,area legend,
error bars/.cd, y dir=both, y explicit]
           coordinates {
          (\LARGE{$C^s_{20}$}, 183.7) += (0,26.97) -= (0,26.97)
          (\LARGE{$C^s_{25}$},190.4) += (0,24.67) -= (0,24.67)
          (\LARGE{$C^s_{30}$},195.32) += (0,27.65) -= (0,27.65)
          
          };

   \addplot [fill=orange,area legend,
error bars/.cd, y dir=both, y explicit]
           coordinates {
          (\LARGE{$C^s_{20}$},225.89) += (0,25.50) -= (0,25.50)
          (\LARGE{$C^s_{25}$}, 238) += (0,23.23) -= (0,23.23)
          (\LARGE{$C^s_{30}$},241) += (0,24.3) -= (0,24.3)
         
          };

    \addplot [fill=green,area legend,
error bars/.cd, y dir=both, y explicit]
           coordinates {
          (\LARGE{$C^s_{20}$},180.34) += (0,31.50) -= (0,31.50)
          (\LARGE{$C^s_{25}$},186.23) += (0,22.87) -= (0,22.87)
          (\LARGE{$C^s_{30}$},190.23) += (0,27.1) -= (0,27.1)
         
          }; 
           
  \end{axis}
  \end{tikzpicture}}
  \end{subfigure}
  \begin{subfigure}[Time complexity]{
\label{subfig:revtce}
\begin{tikzpicture}[scale=0.7]
      \begin{axis}[
      width  = 0.5*\textwidth,
      height = 8cm,
      major x tick style = transparent,
      ybar=5*\pgflinewidth,
      bar width=12pt,
      symbolic x coords={\LARGE{$C^s_{20}$},\LARGE{$C^s_{25}$},\LARGE{$C^s_{30}$}},
      xtick = data,
      scaled y ticks = false,
      enlarge x limits=0.25,
      ymin=0,
      ylabel= \LARGE{Time($s$)},
      legend entries={\Large{GRRW},\Large{RW},\Large{GRBA},\Large{BA}},
      legend style={at={(0.035,0.9)},anchor= west,legend cell align=left,legend columns=2},
      yticklabel style={/pgf/number format/.cd,fixed,precision=2}
  ]

             \addplot [fill=red,area legend,
error bars/.cd, y dir=both, y explicit]
           coordinates {
          (\LARGE{$C^s_{20}$},8.5) += (0,2.32) -= (0,2.33)
          (\LARGE{$C^s_{25}$},15.31) += (0,2.58) -= (0,2.58)
          (\LARGE{$C^s_{30}$},21.9) += (0,2.31) -= (0,2.31)
         
          };

           \addplot [fill=blue,area legend,
error bars/.cd, y dir=both, y explicit]
           coordinates {
          (\LARGE{$C^s_{20}$}, 2.2) += (0,0.4045) -= (0,0.4045)
          (\LARGE{$C^s_{25}$},2.9) += (0,0.43) -= (0,0.43)
          (\LARGE{$C^s_{30}$},3.7) += (0,0.44) -= (0,0.44)
          
          };

   \addplot [fill=orange,area legend,
error bars/.cd, y dir=both, y explicit]
           coordinates {
          (\LARGE{$C^s_{20}$},7.6) += (0,2.11) -= (0,2.11)
          (\LARGE{$C^s_{25}$}, 14.92) += (0,2.3) -= (0,2.3)
          (\LARGE{$C^s_{30}$},20.3) += (0,2.23) -= (0,2.23)
         
          };

    \addplot [fill=green,area legend,
error bars/.cd, y dir=both, y explicit]
           coordinates {
          (\LARGE{$C^s_{20}$},1.9) += (0,0.2534) -= (0,0.2534)
          (\LARGE{$C^s_{25}$},2.1) += (0,0.32) -= (0,0.32)
          (\LARGE{$C^s_{30}$},2.7) += (0,0.37) -= (0,0.37)
         
          }; 
           
  \end{axis}
  \end{tikzpicture}}
  \end{subfigure}  
\caption{Numerical results of the Rev problem in cycle embedding.}
\label{fig:nrRevce}
\end{figure}

\section{Conclusions}
\label{sec:conclusion}
In this work, we systematically investigated the VNE problems in path and cycle topologies. For path embedding, we proved its $\mathcal{NP}$-hardness and inapproximability. Following the idea of expanding substrate networks into "paths", we further developed the MKP-MDKP-based algorithms for the path embedding, which turn out to be more efficient and effective than its counterparts. Regarding cycle embedding, we proposed an auxiliary graph WDAG, based on which we are able to characterize the one-to-one relation between a directed cycle in WDAG and a feasible simplex cycle embedding. Herein is devised a polynomial-time algorithm C2CE to obtain the optimal least-resource-consumption embedding solution. 

As the future work, to realize the idea of decomposing general VNRs into special topologies to efficiently tackle the general VNE problem, there is still a need for dealing with communications between multiple InPs. 
\section*{Appendix}
In this appendix, we present the proof of Lemma \ref{lem1}.
To this end, we need to prove that one of the two problems can be reduced to the other in polynomial time. 

First, we prove $SSET \leq^P_T SG$: Given a graph $G(V,E)$, for each vertex pair $v,u \in V$, we construct a graph $G_{vu}$ by adding a new vertex $v^*$ and connecting  $v^*v$ and $v^*u$.  Afterwards, we obtain a set $\mathcal{G}=\{G_{vu}|\forall v,u \in V\}$ and $|\mathcal{G}|=\binom{|V|}{2}$. It is easy to see that $G$ contains a spanning subgraph which has an Eulerian trail iff there exists one $G_{vu} \in \mathcal{G}$ which is a Supereulerian graph.

Then for $SG \leq^P_T SSET$: Given a graph $G(V,E)$, arbitrarily selecting one vertex $v \in V$, we construct a graph $G^*$ by adding two new vertices $u_1$ and $u_2$ and connecting them to $v$. It is easy to see that $G$ is a Supereulerian graph iff $G^*$ contains a spanning subgraph which has an Eulerian trail.

\bibliographystyle{IEEEtran}
\bibliography{reference}

\end{document}